\documentclass[letterpaper,10pt,conference]{ieeeconf}
\IEEEoverridecommandlockouts
\usepackage[english]{babel}
\usepackage{amsmath,amssymb,amsfonts,mathptmx}
\usepackage{xcolor,graphicx,subfigure}
\usepackage{epsfig}

\let\labelindent\relax
\usepackage{enumitem}
\renewcommand{\theenumi}{(\roman{enumi})}

\catcode`\<=\active \def<{
        \fontencoding{T1}\selectfont\symbol{60}\fontencoding{\encodingdefault}}
\newcommand{\nin}{\not\in}

\newcommand{\tmtextit}[1]{{\itshape{#1}}}

\newtheorem{theorem}{Theorem}[section]

\newtheorem{assumption}[theorem]{Assumption}

\newtheorem{lemma}[theorem]{Lemma}

\newtheorem{remark}[theorem]{Remark}
\newtheorem{example}[theorem]{Example}

\newtheorem{proposition}[theorem]{Proposition}

\newcommand{\oprocendsymbol}{\hbox{$\bullet$}}
\newcommand{\oprocend}{\relax\ifmmode\else\unskip\hfill\fi\oprocendsymbol}

\newcommand{\real}{{\mathbb{R}}}
\newcommand{\naturals}{{\mathbb{N}}}

\newcommand{\GG}{\mathcal{G}}
\newcommand{\II}[1]{\mathcal{I}^{#1}}
\newcommand{\RR}{\mathbf{R}}
\newcommand{\zeros}{\bold{0}}
\newcommand{\ones}{\bold{1}}

\newcommand{\longthmtitle}[1]{\tmtextit{(#1).}}
\newcommand{\myclearpage}{\clearpage}
\renewcommand{\myclearpage}{}

\DeclareMathOperator{\sgn}{sgn}
\DeclareMathOperator{\sat}{sat}
\usepackage[normalem]{ulem}

\renewcommand{\baselinestretch}{.98}

\myclearpage
\begin{document}
        
\title{\Large \bf Double-layered distributed transient frequency
  control with regional coordination for power networks}

\author{Yifu Zhang and Jorge Cort{\'e}s\thanks{The authors are with
    the Department of Mechanical and Aerospace Engineering, University
    of California, San Diego, CA 92093, USA, {\tt\small
      \{yifuzhang,cortes\}@ucsd.edu}}}
\maketitle

\begin{abstract}
  This paper proposes a control strategy for power systems with a
  two-layer structure that achieves global stabilization and, at the
  same time, delimits the transient frequencies of targeted buses to a
  desired safe interval.  The first layer is a model predictive
  control that, in a receding horizon fashion, optimally allocates the
  power resources while softly respecting transient frequency
  constraints.  As the first layer control requires solving an
  optimization problem online, it only periodically samples the system
  state and updates its action. The second layer control, however, is
  implemented in real time, assisting the first layer to achieve
  frequency invariance and attractivity requirements. We show that the
  controllers designed at both layers are Lipschitz in the
  state. Furthermore, through network partition, they can be
  implemented in a distributed fashion, only requiring system
  information from neighboring partitions.  Simulations on the IEEE
  39-bus network illustrate our results.
\end{abstract}

% \begin{keywords}
% Transient frequency, power network stability, distributed control, model %predictive control.
% \end{keywords}

\section{Introduction}\label{section:intro}
Power network frequency is used as a key performance metric in
designing load shedding scheme~\cite{NWM-KC-MS:11}.  In simulations,
such a frequency refers to the system frequency that reflects the
weighted average frequencies of all synchronous generators; however,
in practice, due to the lack of availability of measurements for all
generators, only a few of them are selected and sampled for monitoring
and control design~\cite{FM-FD-GH-DJH-GV:2018}.  Furthermore, from the
point of view of contingency recovery, even if the power supply and
demand are re-balanced after a failure, due to the interconnected
dynamics and inertia of power networks, individual buses may still be
isolated from the network due to overheating relay
protection. Therefore, there is a need of designing control schemes to
restrict single bus transient frequency to evolve within an allowable
range under disturbances and contingencies. This is the problem we
address in this paper, paying special attention to the distributed
implementation of the controller as well as the reduction of the
control effort through cooperation.

% Power network transient stability refers to the ability of an electric
% power network to remain synchronization after being subject to
% disturbances, during which system states should stay within a safe
% bound so that  the entire system is physically
% intact~\cite{PK-JP:04}. Due to the dynamics and interconnection nature
% of power networks, even the power supply and consumption are
% re-balanced immediately after a failure, individual generator is still
% in the danger of overheating caused by large transient frequency or
% voltage deviation, which may in turn trigger cascading failures. From  the practical perspective, it is also common to treat transient frequencies of some crucial generators as key performance metric for evaluating system performance, or, as indexes for load-shedding strategy~\cite{NWM-KC-MS:11}.
% %, e.g., generator loss, load-shedding, and wind energy penetration~\cite{NWM-KC-MS:11}.}
%  These
% motivate us to design a frequency controller aiming to mitigate
% overly high frequency overshoot observed in transients, and at the
% same time, to preserve synchronization of the whole system.

% \marginJC{How different/novel is this lit review with respect to
%   [9,10]? Have we introduced more aspects/considerations/areas that we
% did not touch before? We want to have that, not to repeat/copy ourselves.}
% \marginy{work [5, 6]  is new.}

\textit{Literature review:} Literature~\cite{HDC:11,FD-MC-FB:13}
proposes several sufficient conditions on power network
synchronization; however, as they do not consider bus transient
frequency limit as a constraint, the ideal synchronization condition
may not hold due to possible violation in frequency transients.
Work~\cite{DL-LA-TLV-KT:18} studies the relation between power
injection disturbance and frequency overshoot of individual bus
without active control to regulate frequency transients. On the other
hand, to actively control power network transients, several strategies
have been investigated, including inertial
placement~\cite{TSB-TL-DJH:15}, power system stabilizer~\cite{PK:94},
and power supply re-allocation~\cite{AA-EBM:06}. Yet, these
strategies, aiming at improving system transient behaviors, cannot
rigorously constrain the evolution of frequency to stay within a safe
region. In this regard, we propose two different control
frameworks~\cite{YZ-JC:18-cdc1,YZ-JC:18-cdc2} to achieve both
synchronization and frequency safety. Specifically, as apposed to that in~\cite{YZ-JC:18-cdc1},
 control strategy in~\cite{YZ-JC:18-cdc2} enables cooperation among neighboring buses and reduces the overall control effort in a receding horizon fashion through solving an optimization problem to seek for the optimal control trajectory. However, this
ideal control framework faces practical challenges from two
aspects in real-time implementation. First, it generally takes a long period of time to find the
optimal control trajectory. Second, the control framework requires
finding such an optimal control trajectory at every time instant. 
These deficiencies motivate us to design another framework that can  be implemented in real-time while maintaining the advantage of cooperation.

\textit{Statement of contribution:} 
%
% \marginJC{No references in the statement of contribution. Discuss
%   them, including our own work, in the lit review, point to what they
%   do, and what they cannot achieve to motivate the developments of
%   this paper.}
This paper proposes a control strategy that achieves the following
requirements through a dynamical state-feedback control design:
(i)~The closed-loop system is asymptotically stable. (ii)~For every
targeted bus, under perturbation from power injections or network
dynamical interactions, its whole frequency trajectory stays within a
given safe region, provided its initial frequency lies in the same
region. (iii)~If this is not the case, then the frequency trajectory
should enter the safe region within a finite time and never leaves it
afterwards. (iv)~The control strategy is distributed by only requiring
local state and network information.  Hereby, we propose a
double-layered control structure, where the second layer control
strategy is similar to that in~\cite{YZ-JC:18-cdc2}; however, by
relaxing the frequency constraints and restricting the possible
control trajectory from arbitrary to constant signal, the second-layer
controller only needs to periodically (as opposed to continuously)
solve an optimal control trajectory, and the time consumption for
seeking the optimal one is greatly reduced and almost negligible. The
first layer controller, coming from~\cite{YZ-JC:18-cdc1}, only
slightly tunes the output of the second layer control signal so that
the overall signal rigorously ensures requirement (i)-(iv). We also
show that the proposed control is Lipschitz in state and continuous in
time. We verify our results on the IEEE 39-bus power network.

\section{Preliminaries}\label{section:prelimiaries}
We introduce here notation and notions from graph theory.
               
\emph{Notation:} Let $\naturals$, $\real$, $\real_{>}$, and
$\real_{\geqslant}$ denote the set of natural, real, positive real,
and nonnegative real numbers, respectively.  Variables are assumed to
belong to the Euclidean space if not specified otherwise. Denote
$\ones_n$ and $\zeros_n$ in $\real^n$ as the vector of all ones and
zeros, resp. For $a\in\real$, $\lceil a \rceil$ denote its ceiling. We
let $\|\cdot\|$ denote the 2-norm on $\real^{n}$. For a vector
$b\in\real^{n}$, $b_{i}$ denotes its $i$th entry. For
$A\in\mathbb{R}^{m\times n}$, let $[A]_i$ and $[A]_{i,j}$ denote its
$i$th row and $(i,j)$th element, resp. For any $c,d\in\naturals$, let
$[c,d]_{\naturals}= \left\{ x\in\naturals \big| c\leqslant x\leqslant
  d \right\}$.  
%
% \marginJC{I think you mis-interpreted my ealier margin about
%   $\RR$. This $\RR$ thing is lost here, I'd erase it. When we need to
%   refer to the optimization problem, I'd rather use the equation
%   number and label.}
%   \marginy{I'd keep the $\RR$ notation in the paper, since, for instance, in Lemma 4.2, I need to point out the dependence of  $\RR$ on topology,initial states, etc.}
%
Denote the sign function $\sgn:\real\rightarrow\{0,1\}$ as $ \sgn(a)=
1$ if $a\geqslant 0$, and as $ \sgn(a)= -1$ if $a< 0$.  Finally,
denote the saturation function $\sat:\real \rightarrow \real$ with
limits $a^{\min}<a^{\max}$ by $ \sat(a;a^{\max},a^{\min})= a^{\max}$
if $a\geqslant a^{\max}$, $ \sat(a;a^{\max},a^{\min})= a^{\min}$ if
$a\leqslant a^{\min}$, and $ \sat(a;a^{\max},a^{\min})= a$ otherwise.

\emph{Algebraic graph theory:} We employ basic notions in algebraic
graph theory, cf.~\cite{FB-JC-SM:08cor,NB:94}. An undirected graph is
a pair $\mathcal{G} = \mathcal(\mathcal{I},\mathcal{E})$, where
$\mathcal{I} = \{1,\dots,n\}$ is the vertex set and
$\mathcal{E}=\{e_{1},\dots, e_{m}\} \subseteq \mathcal{I} \times
\mathcal{I}$ is the edge set.  An induced subgraph
$\mathcal{G}_{\beta} = (\mathcal{I}_{\beta},\mathcal{E}_{\beta})$ of
$\mathcal{G} = \mathcal(\mathcal{I},\mathcal{E})$ satisfies
$\mathcal{I}_{\beta}\subseteq\mathcal{I}$,
$\mathcal{E}_{\beta}\subseteq\mathcal{E}$, and
$(i,j)\in\mathcal{E}_{\beta}$ if $(i,j)\in\mathcal{E}$ with
$i,j\in\mathcal{I}_{\beta}$.  Additionally, $\mathcal{E}_{\beta}'
\subseteq\mathcal{I}_{\beta}\times(\mathcal{I}
\backslash\mathcal{I}_{\beta})$ denotes the collection of edges
connecting $\mathcal{G}_{\beta}$ and the rest of the network.
%Further denote
%$\mathcal{I}'_{\beta}\subseteq\mathcal{I}$ as the collection of nodes
%that either belong to $\mathcal{I}_{\beta}$, or is a vertex of some
%edge in $\mathcal{E}'_{\beta}$, i.e., $i\in\mathcal{I}'_{\beta}$ if
%$i\in\mathcal{I}_{\beta}$ or $(i,j)\in\mathcal{E}'_{\beta}$ for some
%$j\in\mathcal{I}_{\beta}$.
A path is an ordered sequence of vertices such that any pair of
consecutive vertices in the sequence is an edge of the graph. A graph
is connected if there exists a path between any two vertices. Two
nodes are neighbors if there exists an edge linking them. Denote
$\mathcal{N}(i)$ as the set of neighbors of node~$i$.  For each edge
$e_{k} \in \mathcal{E}$ with vertices $i,j$, an orientation consists
of choosing either $i$ or $j$ to be the positive end of $e_{k}$ and
the other vertex to be the negative end. The incidence matrix $D=(d_{k
  i}) \in \mathbb{R}^{m \times n}$ associated with $\mathcal{G}$ is
defined as $ d_{k i} = 1$ if $i$ is the positive end of $e_{k}$, $d_{k
  i} = -1$ if $i$ is the negative end of $e_{k}$, and $d_{k i} = 0$
otherwise.

\section{Problem statement}\label{section:ps}
In this section we introduce the dynamics of the power network and the
control requirements.

\subsection{Power network model}\label{subsection:model}
The power network is modeled by a connected undirected graph
$\mathcal{G}=(\mathcal{I},\mathcal{E})$, where
$\mathcal{I}=\{1,2,\cdots,n\}$ stands for the collection of buses
(nodes) and
$\mathcal{E}=\{e_{1},e_{2},\cdots,e_{m}\}\subseteq\mathcal{I}\times\mathcal{I}$
represents the collection of transmission lines (edges).  For every
bus $i\in\mathcal{I}$, let $\omega_{i}\in\real$, $p_{i}\in\real$,
$M_{i}\in\real_{\geqslant }$, and $E_{i}\in\real_{\geqslant }$ denote
the nodal information of shifted voltage frequency relative to the
nominal frequency, active power injection, inertial, and damping
coefficient, respectively. For simplicity, we assume that the latter
two are strictly positive.  Given an arbitrary orientation on
$\mathcal{G}$, for any edge with positive end $i$ and negative end
$j$, let $f_{ij}$ be its signed power flow and $b_{ij}\in\real_{>}$
the line susceptance. Let $\II{u}\subset\mathcal{I}$ be the collection
of buses with exogenous control inputs. To stack this notation in a
more compact way, let $f\in\real^{m}$, $\omega\in\real^{n}$ and
$p\in\real^{n}$ denote the collection of $f_{ij}$'s, $\omega_{i}$'s,
and $p_{i}$'s, resp. Let $Y_{b}\in\real^{m\times m}$ be the diagonal
matrix whose $k$th diagonal entry is the susceptance of the
transmission line $e_{k}$ connecting $i$ and $j$, i.e.,
$[Y_{b}]_{k,k}=b_{ij}$.  Let $M \triangleq
\text{diag}(M_{1},M_{2},\cdots,M_{n})\in\real^{n\times n}$, $E
\triangleq \text{diag}(E_{1},E_{2},\cdots,E_{n})\in\real^{n\times n}$,
and $D\in\real^{m\times n}$ be the incidence matrix.  The linearized
 network dynamics is~\cite{ARB-DJH:81,AP:12},
\begin{subequations}\label{sube:}\label{eqn:compact-form}
  \begin{align}
    \dot f(t)&=Y_{b}D\omega(t),
    \\
    M\dot\omega(t)&=-E\omega(t)-D^{T}f(t)+p(t)+\alpha(t),\label{eqn:compact-form-2}
  \end{align}
\end{subequations}
where $ \alpha(t)\in\mathbb{A}\triangleq\left\{ y\in\real^{n} \big| \
  y_{w}=0 \text{ for }w\in \mathcal{I} \setminus \II{u}\right\}$.  For
convenience, we use $ x\triangleq (f,\omega)\in\real^{m+n}$.
% \begin{align*}
%   A\triangleq \begin{bmatrix} I_{m}& TY_{b}D
%     \\
%     -TM^{-1}D^{T} & I_{n}-M^{-1}TE
%   \end{bmatrix},\ B\triangleq\begin{bmatrix} \zeros_{m}
%     \\
%     TM^{-1}
%   \end{bmatrix},
% % \end{align*}
% \begin{align*}
%   u\in\mathbb{U}\triangleq\left\{ u\in\real^{n}   \big|
%     \ \forall w\in[1,n]_{\naturals},\ [u]_{w}=\left\{ \hspace{-.5cm}\begin{array}{ccc}
%         & u_{w} & \text{if $w\in\II{u}$,}
%         \\
%         & 0 & \text{otherwise}
%       \end{array} \right.    \right\}.
% \end{align*}
We adopt the following assumption on the power injections.

\begin{assumption}\longthmtitle{Finite-time convergence of active
    power injection}\label{assumption:finite-convergence}
  For each $i\in\mathcal{I}$, $p_{i}$ is piece-wise continuous and
  becomes constant (denoted by $p_{i}^{*}$) after a finite time, i.e.,
  there exists $0\leqslant \bar t<\infty$ such that
  $p_{i}(t)=p_{i}^{*}$ for every $i\in\mathcal{I}$ and every
  $t\geqslant \bar t$. Furthermore, the constant power injections are
  balanced, i.e., $\sum_{i\in\mathcal{F}}p_{i}^{*}=0$.
\end{assumption}

Note that Assumption~\ref{assumption:finite-convergence} generalizes
the power injection profile from the commonly used time-invariant case
(e.g.~\cite{CZ-UT-NL-SL:14,ARB-VV:00}) to the finite-time convergent
case. Also, as our controller design here lies in the scope of primary
and secondary control, we assume that the power injection designed by
the tertiary control through economic dispatch is balanced after a
finite time.  Under Assumption~\ref{assumption:finite-convergence},
one can show~\cite{YZ-JC:18-cdc1} that, for the open-loop system
(i.e.,~\eqref{eqn:compact-form} with $\alpha\equiv\zeros_{n}$), the
trajectories $(f(t),\omega(t))$ globally converges to the unique
equilibrium point $(f_{\infty},\zeros_{n})$, where $f_{\infty}$ is
uniquely determined by the power injection profile and network
parameters.
        
\subsection{Control requirements}\label{control-goal}

Our goal is to design distributed state-feedback controllers, one per
each bus $i\in\II{u}$, which maintain stability of the power network
while at the same time cooperatively guaranteeing frequency invariance
and attractivity of nodes in a targeted subset $\II{\omega}$ of
$\II{u}$. Formally, the designed closed-loop system should meet the
following requirements.
\begin{enumerate}[wide]
\item \emph{Frequency invariance:} For each $i\in\II{\omega}$, let
  $\underline\omega_{i}\in\real$ and $\bar\omega_{i}\in\real$ be lower
  and upper safe frequency bounds, with
  $\underline\omega_{i}<\bar\omega_{i}$.  The trajectory of
  $\omega_{i}$ must stay inside
  $[\underline\omega_{i},\bar\omega_{i}]$, provided that its initial
  frequency $\omega_{i}(0)$ lies inside
  $[\underline\omega_{i},\bar\omega_{i}]$. This requirement guarantees
  that every targeted frequency always evolves inside the safe region.

\item \emph{Frequency attractivity:} For each $i\in\II{\omega}$, if
  $\omega_{i}(0)\nin[\underline\omega_{i},\bar\omega_{i}]$, then there
  exists a finite time $t_{0}$ such that
  $\omega_{i}(t)\in[\underline\omega_{i},\bar\omega_{i}]$ for every
  $t\geqslant t_{0}$. This requirement guarantees safe recovery from an
  undesired initial frequency.

\item \emph{Asymptotic stability:} The controller should only regulate
  the system's transients, i.e., the closed-loop system should
  globally converge to the same equilibrium point
  $(f_{\infty},\zeros_{n})$ of the open-loop system.
  
\item \emph{Lipschitz continuity:} The controller must have Lipschitz
  in its state argument.  This suffices to ensure the existence and
  uniqueness of solution for the closed-loop system and, furthermore,
  guarantees that the control action is robust to state measurement
  errors.

\item \emph{Economic cooperation:} The individual controllers
  $\alpha_{i}$, $i\in\II{u}$, should cooperate with each other to
  reduce the overall control effort measure by two norm.
%
%  \marginJC{This requirement is a bit loose if we don't specify how we
%    measure control effort: two norm?}
  % 

\item \emph{Distributed nature:} Every individual controller can only
  utilize the state and power injection information within a local
  region designed by operator. This reflects a practical requirement
  for implementation in larger-scale power networks, in which case
  centralized control strategies depending on global information may
  face critical challenge for real-time execution.
\end{enumerate}

% Our proposed controller possesses a two layer structure which
% together meet the above six requirements. The first layer controller
% periodically and optimally allocates control effort within a local
% region, while respecting an identified stability control and roughly
% adjusting the frequency trajectories as a first step to achieve
% frequency invariance and attractivity. The second layer controller,
% implemented in real-time, slightly tunes the control trajectory
% generated by the first layer, finally ensure frequency invariance
% and attractivity. We then show that both of the two layer controller
% are Lipschitz and only require local information.

\section{Centralized double-layered controller}\label{section:centralized-control}

In this section, we introduce a centralized controller that achieves
the requirements (i)-(v) identified in Section~\ref{control-goal}.
Based on this design, we later propose a distributed version that also
achieves requirement~(vi).

\begin{figure}[tbh]
  \centering%
  \includegraphics[width=1\linewidth]{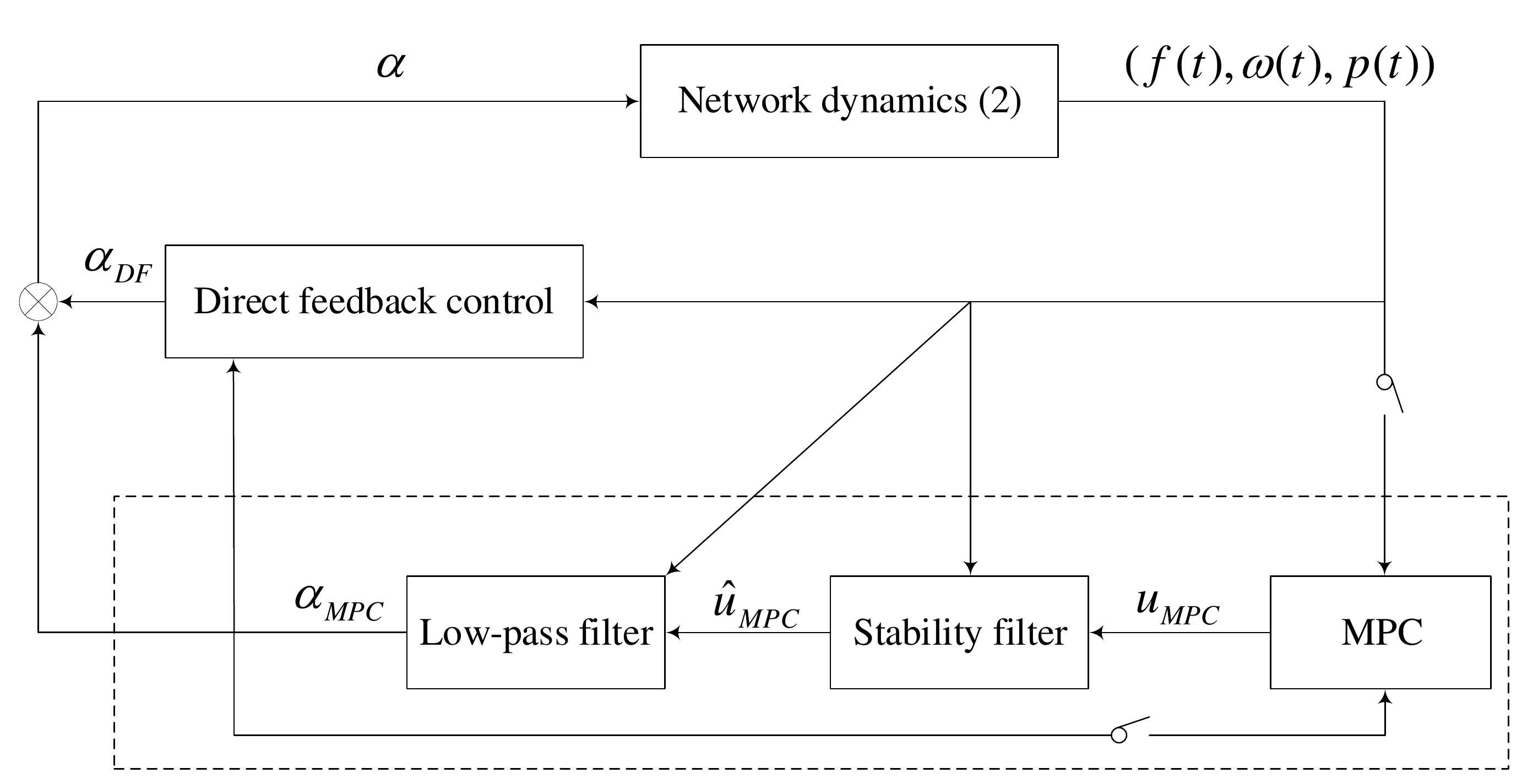}
  \caption{Block diagram of the closed-loop system.}\label{fig:block-diagram}
\end{figure}
We adopt the centralized control structure depicted in
Figure~\ref{fig:block-diagram}.  The control signal $\alpha$ consists
of two parts
\begin{align}\label{eqn:two-layer}
  \alpha = \alpha_{DF}+\alpha_{MPC}.
\end{align}
We next describe the role played by each part. The bottom layer solves
an optimization problem online. To do so, it combines an MPC component
cascaded with a stability filter and a low-pass filter. The MPC
component periodically and optimally allocates control resources,
while roughly adjusting the frequency trajectories as a first step to
achieve frequency invariance and attractivity.  Its output is designed
to be a piece-wise constant signal $u_{MPC}$, which becomes a
piece-wise continuous signal $\hat u_{MPC}$ after passing through
the stability filter. The low-pass filter ensures that the output
$\alpha_{MPC}$ of the bottom layer control is continuous in time to
avoid any discontinuous change in control signal. Using real-time
state information, the stability filter guarantees that $\alpha_{MPC}$
does not jeopardize system stability.  The bottom layer controller
achieves economic cooperation and stabilization, but does not
guarantee frequency invariance and attractivity. The top layer
controller is called direct feedback control since, unlike the bottom
layer control, can be directly computed in real time. This layer
slightly modifies the control generated by the bottom layer to ensure
frequency invariance and attractivity while maintaining stability of
the system.
        
\subsection{Bottom layer controller design via MPC and
  filters}\label{subsection:option-loop}

Here we formally describe each component in the bottom layer control
and analyze their properties.

\subsubsection{MPC component}
The MPC component operates on a periodic time schedule.  In each
sampling period, the MPC component aims to allocate control resources
over controlled nodes in an open-loop fashion based on the latest
sampling system state and forecasted power injection. Here, due to the
additional dynamics of the low-pass filter, the system state consists
of not only power network state $(f,\omega)$, but also the state of
low pass filter $\alpha_{MPC}$ that we later explain.  Formally, let
$\{\Delta^{j}\}_{j\in\naturals}$ be the collection of sampling
points. At time $t=\Delta^{j}$, let a piece-wise continuous signal
$p^{fcst}_{t}:[t,t+\tilde t]\rightarrow\real^{n}$ be the forecasted
value of the power injection $p$ for the first $\tilde t$ seconds
after~$t$.
%
% \marginJC{Is there a requirement that the forecasted value be correct
%   at $t$? }
%   \marginy{No. it doesn't affect invariance, attractivity, or stability, but it affects economic cooperation.}
%
 We discretize the
dynamics~\eqref{eqn:compact-form} and denote $N\triangleq \lceil\tilde
t/T\rceil$ as the length of the predicted step with some $T>0$. At
each $t=\Delta^{j}$, the MPC component updates its output by solving
the following optimization problem,
\begin{subequations}\label{opti:nonlinear}
  \begin{alignat}{2}
    & & & \min_{\hat F,\hat\Omega,\hat A, \hat u,\beta}\quad g(\hat
    u,\beta)\triangleq \sum_{i\in\II{u}}c_{i}\hat
    u^{2}_{i}+d\beta^{2}\notag
    \\
    &\text{s.t.}&\quad &\hat f(k+1)=\hat
    f(k)+TY_{b}D\hat\omega(k),\notag
    \\
    &&&M\hat\omega(k+1)=M\hat\omega(k)+T\big\{-E\hat\omega(k)-D^{T}\hat
    f(k)+\notag
    \\
    &&&\hspace{1.7cm}\hat p^{fcst}(k)+\hat u\big\},\hspace{0cm}\;
    \forall k\in[0,N-1]_{\naturals},\label{opti:nonlinear-1}
    \\
    &&& \hat
    \alpha_{i}(k+1)=\hat\alpha_{i}(k)+T\{-\hat\alpha_{i}(k)/T_{i}-\hat\omega_{i}(k)+\hat
    u_{i}\},\forall i\in\II{u},\notag
    \\
    &&& \hat\alpha_{i}\equiv0,\quad \forall
    i\in\mathcal{I}\backslash\II{u},\label{opti:filter}
    \\
    &&&\hat u\in\mathbb{A},\label{opti:control-location}
    \\
    &&&\hspace{-0.5cm}\hat f(0)=f(\Delta^{j}),\
    \hat\omega(0)=\omega(\Delta^{j}),\
    \hat\alpha(0)=\alpha_{MPC}(\Delta^{j}),\label{opti:nonlinear-2}
    \\
    &&& \hspace{-0.9cm}\underline\omega_{i}-\beta\leqslant
    \hat\omega_{i}(k+1)\leqslant \bar\omega_{i}+\beta, \ \forall
    i\in\II{\omega}\hspace{-.1cm},\forall
    k\in[0,N-1]_{\naturals},\hspace{-0.2cm}\label{opti:nonlinear-3}
    \\
    &&&
    % 
    % -\sgn(\alpha_{i}(\Delta^{j}))\underline\epsilon_{i}\alpha_{i}(\Delta^{j})\leqslant\sgn(\alpha_{i}(\Delta^{j}))\hat %u_{i}\leqslant\notag
    % \\
    % &&&\hspace{2.2cm}\sgn(\alpha_{i}(\Delta^{j}))\bar\epsilon_{i}\alpha_{i}(\Delta^{j}),\quad %\forall i\in\II{u}.\label{opti:sensitivity}
    % \\
    % &&&
    |\hat u_{i}|\leqslant\epsilon_{i}|\alpha_{MPC,i}(\Delta^{j})|,\quad
    \forall i\in\II{u}.\label{opti:sensitivity}
  \end{alignat}
\end{subequations}
%
% \marginJC{The very first time this open-loop optimization is set,
%   $\alpha_{MPC}$ is still not defined, no? I.e., we're doing all of
%   this to define it. So what value should we use in (3f)?}
%
In this optimization problem,~\eqref{opti:nonlinear-1} is the
discretized dynamics corresponding to~\eqref{eqn:compact-form} via
first-order discretization, and $\hat p^{fcst}(k)\triangleq
p^{fcst}_{\Delta^{j}}(\Delta^{j}+kT)$ for every
$k\in[0,N-1]_{\naturals}$;~\eqref{opti:filter} is the discretized
dynamics of the low-pass filter (explained below), with $T_{i}>0$
determining the filter bandwidth;~\eqref{opti:control-location}
indicates the availability of control signal
indexes;~\eqref{opti:nonlinear-2} is the initial state, where
$f(\Delta^{j}),\ \omega(\Delta^{j}),$ and $\alpha_{MPC}(\Delta^{j})$
are sampled state values at time
$t=\Delta^{j}$;~\eqref{opti:nonlinear-3} represents the relaxed
constraint on frequency invariance, where we allow the discretized
frequency $\hat\omega_{i}$ with $i\in\II{\omega}$ exceed its bounds
$\underline\omega_{i}$ and $\bar\omega_{i}$ at the cost of a penalty
term $\beta$;~\eqref{opti:sensitivity} bounds the control input $\hat
u_{i}$ via a coefficient $\epsilon_{i}>0$ as a function of the state
of the low-pass filter to limit the sensitivity to changes in the
latter; the cost function $g$ consists of the overall control effort
as well as a penalty term for frequency violation, where $c_{i}>0$ for
each $i\in\II{u}$ and $d>0$. In the above expression, we use the
compact notation
\begin{subequations}\label{sube:eqn:traj}
  \begin{align}
    \hat F&\triangleq[\hat f(0),\hat f(1),\cdots,f(N)],
    \\
    \hat \Omega&\triangleq[\hat\omega(0),\hat\omega(1),\cdots,\hat\omega(N)],
    \\
    \hat A&\triangleq[\hat \alpha(0),\hat \alpha(1),\cdots,\hat \alpha(N)],
    \\
    \hat P^{fcst}&\triangleq[\hat p^{fcst}(0),\hat
    p^{fcst}(1),\cdots,\hat p^{fcst}(N-1)],\label{eqn:p-G-graph}
  \end{align}
\end{subequations}
as the collection of discretized state trajectories of flow,
frequency, low-pass filter, and forecasted power injection.

%\marginJC{Do we need to introduce the notation $\RR$ given that we can
%  refer to the optimization with the label~\eqref{opti:nonlinear}?}

We refer to the optimization problem~\eqref{opti:nonlinear} as
$\RR(\mathcal{G},\II{u},\II{\omega},p^{fcst}_{\Delta^{j}},
f(\Delta^{j}),\omega(\Delta^{j}),\alpha_{MPC}(\Delta^{j}))$ to
emphasize its dependence on network topology, nodal indexes with
exogenous control signals, nodal indexes with transient frequency
requirement, forecasted power injection, and state values at the
sampling time. If the context is clear, we simply use~$\RR$.  Let
$(\hat F^{*},\hat \Omega^{*},\hat A^{*}, \hat u^{*},\beta^{*})$ denote
its optimal solution. 
% Specially, we may denote $\hat u^{*}$ by $\hat
% u^{*}(\mathcal{G},\II{u},\II{\omega},
% p^{fcst}_{\Delta^{j}},f(\Delta^{j}),\omega(\Delta^{j}),\alpha_{MPC}(\Delta^{j}))$
% to emphasize its dependence on these values.

\begin{remark}\longthmtitle{Selection of frequency violation penalty
    coefficient}\label{rmk:violation-penalty}
  The role of the parameter $d$ in the objective function is to ensure
  that the controller that results from the MPC component does not
  completely disregard the frequency invariance and attractivity
  requirement.
  % plays a fundamental rule in
  % determining how $\hat\omega_{i}^{*}$ with $i\in\II{\omega}$ could
  % exceed the safe region $[\underline\omega_{i},\bar\omega_{i}]$.
  In the extreme case $d=0$ (i.e., no penalty for frequency
  violation), then we have $\hat u^{*} = \zeros_{n}$. As $d$ grows,
  the resulting MPC controller ensures that violations in frequency
  invariance become smaller. The top layer control introduced later
  adds additional input to the resulting controller to ensure the
  frequency requirements.  We come back to this point
  later in the simulations of Section~\ref{sec:simulations}. \oprocend
\end{remark}

Given the open-loop optimization problem~\eqref{opti:nonlinear}, the
function $u_{MPC}$ corresponding to the MPC component in
Figure~\ref{fig:block-diagram} is defined as follows: for
$j\in\naturals$ and $ t\in[\Delta^{j},\Delta^{j+1})$, let
\begin{align}\label{eqn:uMPC}
  u_{MPC}(t) \!=\! \hat
  u^{*}(\mathcal{G},\II{u},\II{\omega}, \hat
  p^{fcst}_{\Delta^{j}},f(\Delta^{j}),\omega(\Delta^{j}),\alpha_{MPC}(\Delta^{j})),
\end{align}
where in the right hand side we emphasize the dependence of $\hat u^{*}$ on the seven arguments.
%
%\marginJC{Why has $t$ disappeared from the right-hand side?}
%\marginy{because $\hat u^{*}$ is designed to be  constant during two succeeding %samplings. That's why $u_{MPC}$ is a piecewise consant signal.}
%
Next, we characterize how the controller depends on the state value
at the sampling time and predicted power injection.
                   %
% \marginJC{Here, and in the lemma title, we say ``on the state value'',
%   but $\hat p^{fcst}_{\Delta^{j}}$ has nothing to do with the state,
%   no?  Also, there is a notational inconsistency, as just in the
%   equation above, $\hat u^*$ depends on 7 arguments, whereas in the
%   lemma below, only depends on 4 arguments. }
%

\begin{lemma}\longthmtitle{Piece-wise affine and continuous dependence
    of optimal solution on sampling state and predicted power injection}\label{lemma:pwl-optimal}
  The optimization problem $\RR(\mathcal{G},\II{u},\II{\omega},
  p^{fcst}_{\Delta^{j}},f(\Delta^{j}),\omega(\Delta^{j}),\alpha_{MPC}(\Delta^{j}))$
  in~\eqref{opti:nonlinear} has a unique optimal solution $(\hat
  F^{*},\hat \Omega^{*},\hat A^{*}, \hat
  u^{*},\beta^{*})$. Furthermore, given $\mathcal{G},\ \II{u}$, and
  $\II{\omega}$, $\hat u^{*}$ is a continuous and piece-wise affine in
  $(\hat
  P^{fcst},f(\Delta^{j}),\omega(\Delta^{j}),\alpha_{MPC}(\Delta^{j}))$,
  %
%  \marginJC{What is $\hat P^{fcst}$? You mean
%    $\hat p^{fcst}_{\Delta^{j}}$? Same comment below.}
%    \marginy{$\hat P^{fcst}$ is defined in~\eqref{eqn:p-G-graph}.}
  %
  that is, there exist $l\in\naturals, \{H_{i}\}_{i=1}^{l}$,
  $\{S_{i}\}_{i=1}^{l}$, $\{h\}_{i=1}^{l}$, and $\{s_{i}\}_{i=1}^{l}$
    %
%   \marginJC{We have already used notation like $F$ and $f$ for flows,
%     so we should use something completely different here to avoid
%     confusing the reader.}
  %
  with suitable dimensions such that
  \begin{align}\label{eqn:pwa-u}
    \hat u^{*}=S_{i}z+s_{i}, \text{ if } z\in\left\{ y
      \big|H_{i}y\leqslant h_{i} \right\}\text{ and }
    i\in[1,l]_{\naturals}
  \end{align}
  %
%  \marginJC{There is no dependency of time? So the input is constant
%    throughout the whole time interval $[\Delta^j,\Delta^{j+1}]$??}
%   \marginy{Yes. It is intentionlly designed to be a constant to reduce the %computational time for solving the MPC.}
  %
  holds for every $z\in\real^{(N+2)n+m}$, where $z$ is the collection
  of $(\hat
  P^{fcst},f(\Delta^{j}),\omega(\Delta^{j}),\alpha_{MPC}(\Delta^{j}))$
  in a column vector form.
\end{lemma}
\begin{proof}
  We start by noting that $\RR$ is feasible (hence at least one
  optimal solution exists) for any given
  $z$. %$(\hat P^{fcst},f(\Delta^{j}),\omega(\Delta^{j}),\alpha(\Delta^{j}))$,
  This is because, given a state trajectory
  $(\hat F,\hat \Omega,\hat A)$
  of~\eqref{opti:nonlinear-1}-\eqref{opti:filter} with input
  $\hat u=\zeros_{n}$ and initial condition~\eqref{opti:nonlinear-2},
  choosing a sufficiently large $\beta$ makes it satisfy
  constraint~\eqref{opti:nonlinear}.  The uniqueness follows from the
  strict convexity of $g$ and the linearity of constraints.  To show
  continuity and piece-wise affinity, we separately consider
  $2^{|\II{u}|}$ cases, depending on the sign of each
  $\{\alpha_{MPC,i}(\Delta^{j})\}_{i\in\II{u}}$. Specifically, let
  $\eta\triangleq\{\eta_{i}\}_{i\in\II{u}}\in\{1,-1\}^{|\II{u}|}$ and
  define
  $\mathfrak{B}^{\eta}\triangleq\left\{ z
    \big|(-1)^{\eta_{i}}\alpha_{MPC,i}(\Delta^{j})\geqslant 0,\
    \forall i\in\II{u} \right\}$.  Note that every $z$ lies in at
  least one of these sets and that, in any $\mathfrak{B}^{\eta}$, the
  sign of each $\alpha_{MPC,i}(\Delta^{j})$ with $i\in\II{u}$ is
  fixed.  Hence all the $|\II{u}|$ constraints
  in~\eqref{opti:sensitivity} can be transformed into one of the
  following forms
  \begin{subequations}\label{sube:ineq:sgn-alpha}
    \begin{alignat}{2}
    \hspace{-0.38cm}  -\epsilon_{i}\alpha_{MPC,i}(\Delta^{j}) & \leqslant \hat
      u_{i}\leqslant\epsilon_{i}\alpha_{MPC,i}(\Delta^{j}) & & \text{if
      }\alpha_{MPC,i}(\Delta^{j})\geqslant 0,\label{sube:ineq:sgn-alpha-1}
      \\
        \hspace{-0.2cm}  \epsilon_{i}\alpha_{MPC,i}(\Delta^{j}) & \leqslant \hat
      u_{i}\leqslant-\epsilon_{i}\alpha_{MPC,i}(\Delta^{j}) && \ 
      \text{if } \alpha_{MPC,i}(\Delta^{j})\leqslant
      0 . \label{sube:ineq:sgn-alpha-2}
    \end{alignat}
  \end{subequations}
  %
%  \marginJC{Should the lefthand side in (7a) have a minus sign?? Also,
%    why $\alpha_{i}(\Delta^{j})$ in the if's instead of
%    $\alpha_{MPC,i}(\Delta^{j})$??}
  %
  Note that if $\alpha_{MPC,i}(\Delta^{j})=0$, then $\hat
  u_{i}=0$. Therefore, in every $\mathfrak{B}^{\eta}$, $z$ appears in
  $\RR$ in a linear fashion; hence,
  % , which, together with the fact that $g$ is strictly convex in
  % $(\hat u,\beta)$, implies the uniqueness of $(\hat
  % u^{*},\beta^{*})$. The uniqueness of $(\hat F^{*},\hat
  % \Omega^{*},\hat
  % A^{*})$ %follows as it is uniquely determined once $\hat u^{*}$ is fixed.
  it is easy to re-write $\RR$ into the following form:
  \begin{alignat}{2}
    &\min_{c} & \quad & c^{T}Kc\notag
    \\
    &\text{s.t.}&\quad &Gc\leqslant W+J^{\eta}z,
  \end{alignat}
  where $c$ is the collection of $(\hat F,\hat\Omega,\hat A, \hat
  u,\beta)$ in vector form and $K\succeq 0$, $G$, $W$ and $J^{\eta}$
  are matrices with suitable dimensions. Note that only $J^{\eta}$
  depends on $\eta$. By~\cite[Theorem~1.12]{FB:03}, for every
  $\eta\in\{-1,1\}^{|\II{u}|}$, $c^{*}$ is a continuous and piece-wise
  affine function of $z$ whenever $z\in\mathfrak{B^{\eta}}$.  Since
  each $\mathfrak{B^{\eta}}$ consists of only linear constraints and
  the union of all $\mathfrak{B^{\eta}}$'s with
  $\eta\in\{1,-1\}^{|\II{u}|}$ is $\real^{(N+2)n+m}$, one has that
  $c^{*}$ is piece-wise affine in $z$ on $\real^{(N+2)n+m}$. Lastly,
  to show the continuous dependence of $c^{*}$ on $z$ on
  $\real^{(N+2)n+m}$, note that since such a dependence holds on every
  closed set $\mathfrak{B^{\eta}}$, we only need to prove that $c^{*}$
  is unique for every $z$ lying on the boundary shared by different
  $\mathfrak{B^{\eta}}$'s. This holds trivially as $c^{*}$ is unique
  for every $z\in\real^{(N+2)n+m}$, which we have proven above.
%    To
%   complete the proof, we only need to prove that $c^{*}$ is continuous
%   at the boundary of every $\mathfrak{B}^{\eta}$, which holds
%   trivially as the boundary between two different
%   $\mathfrak{B}^{\eta}$'s belongs to both regions, representing the
%   same $\RR$.
\end{proof}

Notice that Lemma~\ref{lemma:pwl-optimal} implies that $\hat u^{*}$ is
globally Lipschitz in $z$ (and hence in the sampled state
$f(\Delta^{j}),\omega(\Delta^{j})$, and $\alpha_{MPC}(\Delta^{j})$),
with $L\triangleq\max_{i\in[1,l]_{\naturals}}\|F_{i}\|$ serving as a
global Lipschitz constant. Another interesting consequence of this
result is that
% For each sampling time $\Delta^{j}$ with forecasted power injection
% and sampled state
% $(\hat
% P^{fcst},f(\Delta^{j}),\omega(\Delta^{j}),\alpha_{MPC}(\Delta^{j}))$,
it provides an alternative to directly solving the optimization
problem~$\RR$. In fact, one can compute and store offline
$\{H_{i}\}_{i=1}^{l}$, $\{S_{i}\}_{i=1}^{l}$, $\{h\}_{i=1}^{l}$, and
$\{s_{i}\}_{i=1}^{l}$, and then compute $\hat u^{*}$ online
using~\eqref{eqn:pwa-u}. However, this approach
% Although this seems to be attractive as the
% online computation only involves some simple linear operation to
% compute the control signal, it
faces practical difficulties regarding storage
capacity~\cite{JBR-DQM:09}, as the number $l$ grows exponentially with
system order $m+n$, input size $|\II{u}|$, as well as the horizon
length~$N$.

\subsubsection{Stability and low-pass filter}
Next we introduce the stability and low-pass filters.  Note that for
any time $t\in(\Delta_j,\Delta_{j+1})$, due to the sampling mechanism,
$u_{MPC}(t)$ depends on the old sampled state at time $\Delta_j$, as
opposed to the state information at current time $t$. Since such a
lack of update may jeopardize system stability, we cascade a stability
filter that depends on the current state after the MPC component to
filter out the unstable part in $u_{MPC}$. The goal of low-pass filter
is to simply ensure that the output of the bottom layer is continuous
in time.
%The stability
%filter uses the current system state and the piece-wise constant
%signal $u_{MPC}$ as input, and generates a piece-wise continuous
%signal. 
Formally, for every $i\in\II{u}$ at any $t\geqslant 0$, define the
stability filter as
% every $t\in[\Delta^{j},\Delta^{j+1})$, and every $j\in\naturals$
\begin{align}\label{eqn:stability-filter}
\hspace{-1cm}  \hat u_{MPC,i}(\alpha_{MPC}(t),u_{MPC}(t))
  &\notag
  \\
  &\hspace{-3cm}=\sat(u_{MPC,i}(t);\epsilon_{i}|\alpha_{MPC,i}(t)|,-\epsilon_{i}|\alpha_{MPC,i}(t)|),
% \\
%  \hat u_{MPC,i}=& \notag
%  \\
%    &\hspace{-3.9cm}\begin{cases}
%      0 & \hspace{-1cm}\text{if $\exists\tau\in[\Delta^{j},t]$ s.t. $|u_{MPC,i}(t)|> \epsilon_{i}|\alpha_{i}(\tau)|$,}
%      \\
%      u_{MPC,i}(t) & \hspace{1.2cm}\text{otherwise.}
%    \end{cases}
\end{align}
and define the low-pass filter as                      %
\begin{align}\label{eqn:lp-filter}
  \dot\alpha_{MPC,i}(t)&=-\frac{1}{T_{i}}\alpha_{MPC,i}(t)-\omega_{i}(t)+\hat
  u_{MPC,i}(t),\quad\forall i\in\II{u},\notag
  \\
  \alpha_{MPC,i}&\equiv0,\quad\forall i\in\mathcal{I}\backslash\II{u}.
\end{align}
% \marginJC{Where does $t$ belong to here?
%   $t \in [\Delta_j,\Delta_{j+1}]$? I'm confused about this definition,
%   it seems to be non-causal. I need to know the value of
%   $\alpha_{MPC}(t)$ to define this input, but this input in turn is
%   going to be used to define $\alpha_{MPC}(t)$??? This doesn't look
%   correct.}
% \marginy{We need to know the value of  $\alpha_{MPC}(t)$ to define this %input so as to define $\hat u_{MPC}$, but $u_{MPC}$ uses $\alpha_{MPC}$, %which is the ouput of the low-pass finter,  as input, not \hat u_{MPC}.}
%
% where $\hat u_{MPC,i}$ (resp. $\alpha_{MPC,i}$ and $u_{MPC,i}$) is the
% $i$th component of $\hat u_{MPC}$ (resp. $\alpha_{MPC,i}$ and
% $u_{MPC}$).
% %
% \marginJC{We've used $\alpha_{MPC,i}$ so many times before that it
%   seems odd to introduce now what it means. I think i'd just erase
%   this -- we already say something to this effect in the notation.}
%
Note that the low-pass filter model matches the structure in the
discretized model~\eqref{opti:filter}. Also,
both~\eqref{eqn:stability-filter} and~\eqref{eqn:lp-filter} can be
implemented in a distributed fashion: $\alpha_{MPC,i}$ depends on
$\omega_{i}$ and $\hat u_{MPC,i}$, and $\hat u_{MPC,i}$ only relies on
$\alpha_{MPC,i}$ and $u_{MPC,i}$, both of which are local information
for node $i$. For simplicity, we interchangeably use $\hat
u_{MPC,i}(\alpha_{MPC}(t),u_{MPC}(t))$ and $\hat u_{MPC,i}(t)$.

The following result shows the Lipschitz continuity of $\hat u_{MPC}$
and points out a condition it satisfies. Later we show that this condition
 ensures stability of the closed-loop system.

\begin{lemma}\longthmtitle{Lipschitz continuity and stability
    condition} \label{lemma:Lipschitz-stability}
  For the signal $\hat u_{MPC}$ defined
  in~\eqref{eqn:stability-filter}, $\hat u_{MPC}$ is Lipschitz in
  system state at every sampling time $t=\Delta^{j}$ with
  $j\in\naturals$. Additionally, it holds that
  \begin{align}\label{ineq:stability-condition}
    \alpha_{MPC,i}(t)\hat
    u_{MPC,i}(t)\leqslant\epsilon_{i}\alpha_{MPC,i}^{2}(t),\quad\forall
    t\geqslant 0,\ \forall i\in\mathcal{I}.
  \end{align}
\end{lemma}
\begin{proof}
  If $t=\Delta^{j}$, then since $|\hat
  u_{i}^{*}|\leqslant\epsilon_{i}|\alpha_{MPC,i}(\Delta^{j})|$ and
  $u_{MPC,i}(\Delta^{j})=\hat u_{i}^{*}$ for every $i\in\II{u}$,
  by~\eqref{eqn:stability-filter}, it holds that $\hat
  u_{MPC,i}(\alpha_{MPC}(t),u_{MPC}(t))|_{t=\Delta^{j}}=\hat
  u_{i}^{*}$. The Lipschitz continuity follows by
  Lemma~\ref{lemma:pwl-optimal}.
  %
%   \marginJC{I'm sorry, I don't see why. The function $\alpha_{MPC}(t)$
%     has still not being defined, so I highly doubt we can claim
%     Lispchitz properties of $\hat u_{MPC}$.}
%     \marginy{There is an abuse of notation, but without such an abuse, the notation would be complicated. As shown in proof, $\hat
%   u_{MPC,i}(\alpha_{MPC}(t),u_{MPC}(t))|_{t=\Delta^{j}}=\hat
%   u_{i}^{*}$, and since $\hat u_{i}^{*}$ is Lipschitz, the left hand side is also lipschitz. The abouse of notation comes from the aspect that $u_{MPC}$ can be either treated as a signal of pure time, or a signal of both time and state.}
  %
  Condition~\eqref{ineq:stability-condition}
  simply follows by the definition of saturation function.
\end{proof}

By Lemma~\ref{lemma:Lipschitz-stability}, since $\hat u_{MPC}$ is
Lipschitz at every sampling time, if the top layer controller is also
Lipschitz (which is the case, see Section~\ref{subsection:top layer}),
then the solution of the closed-loop system exists and is unique and
continuous in time. By~\eqref{eqn:stability-filter} and noting that
$u_{MPC}$ is piece-wise constant in time, one has $\hat u_{MPC}$ is
piece-wise continuous, which further implies that $\alpha_{MPC}$ is
indeed continuous in time.
%
% \marginJC{I'm still confused with this non-causality of the
%   control. We are using $t$, but talking about Lispchitz in ``system
%   state'', which does not show up explicitly. }
%   \marginy{See my previous margin a few lines above. I admit there is an %abose of notation.}
%

\begin{remark}\longthmtitle{Almost independent design of MPC component
    and stability filter}\label{rmk:independent-design}
  Note that, regardless of the MPC component output $u_{MPC}$, the
  output of the stability filter $\hat u_{MPC}$ defined
  in~\eqref{eqn:stability-filter} always meets
  condition~\eqref{ineq:stability-condition}.
%
% \marginJC{Try to avoid sentences that are so long. Break them into
%   simpler sentences, so that simpler arguments/points get across to
%   build potentially more complex observations. There are 2 things
%   going on here: that (10) is a stability condition and the
%   independence between MPC and filter. I think the comments about (10)
% should be in the body of the paper, the other can be the remark.}
%
  This strong property provides flexibility in the MPC component design
  and robustness against, e.g., inaccuracy in sampled state
  measurement, forecasted power injection, as well as system
  parameters. However, to ensure the Lipschitz continuity in
  Lemma~\ref{lemma:Lipschitz-stability}, we cast
  constraint~\eqref{opti:sensitivity} that shares a same coefficient
  $\epsilon_{i}$ in the stability
  filter~\eqref{eqn:stability-filter}. To address this aspect, the
  designs of the MPC component and stability filter are not completely
  independent.  \oprocend
\end{remark}

\begin{remark}\longthmtitle{Unnecessity of stability filter with
    continuous sampling MPC component}\label{rmk:unnecessity-filter} 
  Our purpose of considering the stability filter here is to ensure
  that the filtered signal $\hat u_{MPC}$ satisfies
  condition~\eqref{ineq:stability-condition}. However, if the MPC
  component can ideally sample the system state in a continuous
  fashion (as opposed to periodic sampling), then there is no need to
  additionally add such a stability filter, as pre-filtered signal
  already satisfies such a condition.  In detail, if we define
  $u_{MPC}(t)=\hat u^{*}(\mathcal{G},\II{u},\II{\omega},
  p^{fcst}_{t},f(t),\omega(t),\alpha_{MPC}(t))$ for every $t\geqslant
  0$, then due to constraint~\eqref{opti:sensitivity}, one can easily
  see $\alpha_{MPC,i}(t)
  u_{MPC,i}(t)\leqslant\epsilon_{i}\alpha_{MPC,i}^{2}(t), \forall
  t\geqslant 0,\ \forall i\in\mathcal{I}$. Hence, in this scenario,
  one can safely ignore the stability filter and directly have
  $u_{MPC}$ as the input of the low-pass filter.%   Due to practical
  % difficulty of computing MPC output for every time instant, we still
  % focus on the periodic sampling case, where the stability filter is
  % of necessity. 
  \oprocend
\end{remark}

% 
% So far we have introduced the structure of the bottom layer
% controller, consisting of a MPC component to allocate control signals,
% a stability filter to ensure a stability condition, and a low-pass
% filter to make the bottom layer control output continuous in time. We
% have also explored the structural properties of MPC component output
% with respect to sampled state and established Lipschitz continuity.
%                   %
% \marginJC{This paragraph is misplaced here. Better to fuse it/move it
%   to the beginning of this subsection. Also, the discussion as to why
%   the stability filter ``ensures a stability condition'' should be in
%   the body of the paper, not in a remark. }
%

\subsection{Top layer design through direct feedback}\label{subsection:top layer}
%
% \marginJC{This section is too brief and could use some expansion. In
%   particular, we don't explain to the reader the rationale behind the
%   design of~\cite{YZ-JC:18-cdc1}, the role that $\alpha_{MPC}(t)$
%   plays in it, etc. A few sentences along the lines ``the closed-loop
%   system with $\alpha_{MPC}(t)$ might violate frequency invariance,
%   hence we treat it as an input to the controller in [8] that is
%   precisely designed to correct potential violations of the frequency
%   requirements by kicking in as the margin gets smaller '' or
%   something like that would help.}
%

Although the bottom layer control attempts to achieve frequency
invariance and attracitiviy requirements by constraining the predicted
frequency trajectory via~\eqref{opti:nonlinear-3}, it cannot solely
guarantee the two requirements. To address this aspect, we construct
the top layer control from~\cite{YZ-JC:18-cdc1} that is precisely
designed to correct potential violations of the frequency requirements
by kicking in as the margin of violations gets smaller. Formally, for
every $i\in\II{\omega}$, let $\bar\gamma_{i},\underline \gamma_{i}>0$,
and $\underline \omega_{i}^{\text{thr}}, \bar\omega_{i}^{\text{thr}}
\in \real$ with $\underline \omega_{i}< \underline
\omega_{i}^{\text{thr}} < 0 <
\bar\omega_{i}^{\text{thr}}<\bar\omega_{i}$.  Define the top layer
controller $\alpha_{DF}$ as in~\eqref{eqn:second-layer-control}.

Note that the top layer control signal is only available for node with
index in $\II{\omega}$, and that $\alpha_{DF}$ can be implemented
distributedly, in the sense that for each $\alpha_{DF,i}$ with
$i\in\II{\omega}$ regulated at node $i$, it only requires its nodal
frequency $\omega_{i}$, aggregated power flow $[D^{T}]_{i}f$, power
injection $p_{i}$, as well as the local bottom layer control signal
$\alpha_{MPC,i}$. In addition, we have shown~\cite{YZ-JC:18-cdc1} that
$\alpha_{DF}$ is locally Lipschitz in its first argument.  If the
context is clear, we may interchangeably use
$\alpha_{DF,i}(x(t),p(t),\alpha_{MPC}(t))$
(resp. $v_{i}(x(t),\alpha_{MPC}(t),p(t))$) and $\alpha_{DF,i}(t)$
(resp. $v_{i}(t)$).

\begin{figure*}[htb]
\begin{subequations}\label{eqn:second-layer-control}
  \begin{alignat}{2}
    &\forall i\in\II{\omega }, \text{ let }
    \alpha_{DF,i}(x(t),p(t),\alpha_{MPC}(t)) &&=
    \begin{cases}
      \min\{0,\frac{\bar\gamma_{i}(\bar\omega_{i}-\omega_{i}(t))}{\omega_{i}(t)-\bar\omega_{i}^{\text{thr}}}+v_{i}(x(t),\alpha_{MPC}(t),p(t))\}
      & \omega_{i}(t)>\bar\omega_{i}^{\text{thr}},
      \\
      0 & \underline\omega_{i}^{\text{thr}}\leqslant
      \omega_{i}(t)\leqslant \bar\omega_{i}^{\text{thr}},
      \\
      \max\{0,\frac{\underline\gamma_{i}(\underline\omega_{i}-\omega_{i}(t))
      }{\underline\omega_{i}^{\text{thr}}-\omega_{i}(t)}+v_{i}(x(t),\alpha_{MPC}(t),p(t))\}
      & \omega_{i}(t)<\underline\omega_{i}^{\text{thr}},
    \end{cases}
    \\
&\hspace{2cm}v_{i}(x(t),\alpha_{MPC}(t),p(t))&&\triangleq
  E_{i}\omega_{i}(t)+[D^{T}]_{i}f(t)-p_{i}(t)-\alpha_{MPC,i}(t),
  \\
 & \forall i\in\mathcal{I}\backslash\II{\omega},\text{ let }\hspace{1.5cm}\alpha_{DF,i}&&\equiv0.
  \end{alignat}
\end{subequations}
  \hrulefill
\end{figure*}

\subsection{Closed-loop stability, frequency invariance, and frequency
  attractivity analysis}
With both layers introduced, we are ready to analyze the stability of
the closed-loop system and show that it meets the
requirements~(i)-(iii) in Section~\ref{control-goal}. Notice that as
we individually introduce each component in the control scheme, we
have shown that all components are Lipschitz, and the economic
cooperation is encoded in the MPC component; therefore, the
requirements~(iv) and~(v) are met.

\begin{theorem}\longthmtitle{Centralized double-layered control with
    stability and frequency guarantees}\label{thm:two-layer-control}
  Under Assumption~\ref{assumption:finite-convergence}, if
  $\epsilon_{i}T_{i}<1$ for every $i\in\II{u}$, then the
  system~\eqref{eqn:compact-form} with controller defined
  by~\eqref{eqn:two-layer},~\eqref{eqn:uMPC},~\eqref{eqn:stability-filter},~\eqref{eqn:lp-filter},
  and \eqref{eqn:second-layer-control} meets
  requirements~(i)-(iii). Furthermore, $\alpha(t)$, $\alpha_{MPC}(t)$,
  and $\alpha_{DF}(t)$ converge to $\zeros_{n}$ as $t\rightarrow
  \infty$.
%   \begin{enumerate}
%   \item\label{item:convergence} The state
%     $(f(t),\omega(t),\alpha_{MPC}(t))$ globally converges to
%     $(f_{\infty},\zeros_{n},\zeros_{n})$;
%   \item\label{item:invariance} For any $i\in\II{\omega}$, if
%     $\omega_{i}(0)\in[\underline\omega_{i},\bar\omega_{i}]$, then
%     $\omega_{i}(t)\in[\underline\omega_{i},\bar\omega_{i}]$ for every
%     $t>0$;
%   \item\label{item:attracitivity} For any $i\in\II{\omega}$, if
%     $\omega_{i}(0)\nin[\underline\omega_{i},\bar\omega_{i}]$, then
%     there exists $t_{0}>0$ such that
%     $\omega_{i}(t)\in[\underline\omega_{i},\bar\omega_{i}]$ for every
%     $t\geqslant t_{0}$.
%   \end{enumerate}
\end{theorem}
\begin{proof}
  We first consider requirement~(iii). Without loss of generality, we
  can assume that $p$ is constant so that the closed-loop system is
  time-invariant. Otherwise one can simply consider $t=\bar t$ as the
  initial state.
  % 
%   \marginJC{Doesn't the justification take just 1 line? I wouldn't
%     mysteriously refer to another paper for it. In fact, I'd try to
%     make the exposition as self-contained as possible.  We can use
%     results of other papers, including ours, but this incomplete
%     referencing inside a proof for facts that should be articulated
%     here is not good.}
  % 
 Select the energy function
  \begin{align*}
    V(f,\omega,\alpha_{MPC})\triangleq\frac{1}{2}(f-f_{\infty})^{T}(f-f_{\infty})
    + \frac{1}{2}\omega^{T} M \omega +
    \frac{1}{2}\alpha_{MPC}^{T}\alpha_{MPC}
  \end{align*}
  After some computations, we obtain
  \begin{align*}
    \dot
    V=&-\omega^{T}(t)E\omega(t)+\sum_{i\in\II{\omega}}\omega_{i}(t)\alpha_{DF,i}(t)
    \\
    &-\sum_{i\in\II{u}}\left(\frac{1}{T_{i}}\alpha_{MPC,i}^{2}(t)-\alpha_{MPC,i}(t)\hat
      u_{MPC,i}(t)\right).
  \end{align*}
  Note that by the definition of $\alpha_{DF}$
  in~\eqref{eqn:second-layer-control},
  $\omega_{i}(t)\alpha_{DF,i}(t)\leqslant0$ holds for every
  $i\in\II{\omega}$ at every $t\geqslant 0$, in that
  $\alpha_{DF,i}(t)=0$ whenever
  $\underline\omega_{i}^{\text{thr}}\leqslant \omega_{i}(t)\leqslant
  \bar\omega_{i}^{\text{thr}}$, and $\alpha_{DF,i}(t)\geqslant 0$
  (reps. $\leqslant 0$) if $\omega_{i}(t)\geqslant
  \bar\omega_{i}^{\text{thr}}>0$ (resp.  $\omega_{i}(t)\leqslant
  \underline\omega_{i}^{\text{thr}}<0$). Therefore, together with
  condition~\eqref{ineq:stability-condition} in
  Lemma~\ref{lemma:Lipschitz-stability}, we have
  \begin{align*}
    \dot V\leqslant-\omega^{T}(t)E\omega(t)-\sum_{i\in\II{u}}
    (\frac{1}{T_{i}}-\epsilon_{i})\alpha_{MPC,i}^{2}(t) \leqslant0.
  \end{align*}
  The convergence follows by LaSalle's invariance
  principle. Specifically, $\omega(t)$ and $\alpha_{MPC}(t)$ converge
  to $\zeros_{n}$ (notice that $\alpha_{MPC,i}\equiv0$ for each
  $i\in\mathcal{I}\backslash\II{u}$). Next we show that
  $\lim_{t\rightarrow\infty}\alpha_{DF,i}(t)=0$ for every
  $i\in\II{\omega}$, which implies that
  $\lim_{t\rightarrow\infty}\alpha_{DF}(t)=\zeros_{n}$ as
  $\alpha_{DF,i}\equiv 0$ for each
  $i\in\mathcal{I}\backslash\II{\omega}$. This simply follows
  from~\eqref{eqn:second-layer-control} since $\alpha_{DF,i}(t)=0$
  whenever
  $\underline\omega_{i}^{\text{thr}}\leqslant\omega_{i}(t)\leqslant
  \bar\omega_{i}^{\text{thr}}$, where $0\in(
  \underline\omega_{i}^{\text{thr}}, \bar\omega_{i}^{\text{thr}})$,
  and we have shown that
  $\lim_{t\rightarrow\infty}\omega(t)=\zeros_{n}$. The convergence of
  $\alpha(t)$ follows by its definition~\eqref{eqn:two-layer}.

%   Now, since for  every $i\in\II{u}$, $\alpha_{DF,i}(t)= 0$  whenever $\underline\omega_{i}^{\text{thr}}\leqslant
%       \omega_{i}(t)\leqslant \bar\omega_{i}^{\text{thr}}$, and $\alpha_{DF,i}\equiv0$ with $i\in\mathcal{I}\backslash\II{u}$, one has that $\alpha_{DF}(t)$ converges
%   , and we
%   ignore its proof as a similar one can be found
%   in~\cite{YZ-JC:18-cdc1}.
%   %
%   \marginJC{Again, not a good idea. It's LaSalle, so no need to refer
%     to our work (we might even put the reviewers on alert, as to what
%     is the technical novelty then in this work). If the reasoning is
%     long, just provide an abbreviated version of it in here.}
%   %

  For requirement~(i), we have shown
  in~\cite{YZ-JC:18-cdc1} that it is equivalent to asking that for
  any $i\in\II{\omega}$ at any $t\geqslant 0$,
  \begin{subequations}\label{eqn:ith-dyanmics}
    \begin{align}
      \dot\omega_{i}(t)\leqslant 0\text{ if
      }\omega_{i}(t)=\bar\omega_{i},\label{eqn:ith-dyanmics-a}
      \\
      \dot\omega_{i}(t)\geqslant 0\text{ if
      }\omega_{i}(t)=\underline\omega_{i}.\label{eqn:ith-dyanmics-b}
    \end{align}
  \end{subequations}
  For simplicity, we only prove~\eqref{eqn:ith-dyanmics-a},
  and~\eqref{eqn:ith-dyanmics-b} follows similarly. Note that
  by~\eqref{eqn:compact-form-2},~\eqref{eqn:two-layer},
  and~\eqref{eqn:second-layer-control}, one has
  \begin{align*}
    \dot\omega_{i}(t)&=-E_{i}\omega_{i}(t)-[D]^{T}f(t)+p_{i}(t)+\alpha_{i}(t)
    \\
    &=-E_{i}\omega_{i}(t)-[D]^{T}f(t)+p_{i}(t)+\alpha_{MPC,i}(t)+\alpha_{DF,i}(t)
    \\
    &=-v_{i}(t)+\alpha_{DF,i}(t).
  \end{align*}
  Now if $\omega_{i}(t)=\bar\omega_{i}$, then
  $-v_{i}(t)+\alpha_{DF,i}(t)=-v_{i}(t)+\min\{0,v_{i}(t)\}\leqslant0$;
  hence condition~\eqref{eqn:ith-dyanmics-a} holds.

  Finally, requirement~(ii) follows immediately
  from~(i) and~(iii). As for any
  $i\in\mathcal{I}$, $\omega_{i}$ converges to
  $0\in(\underline\omega_{i},\bar\omega_{i})$, there must exist a
  finite time $t_{0}$ such that
  $\omega_{i}(t_{0})\in[\underline\omega_{i},\bar\omega_{i}]$, which,
  by frequency invariance, implies that
  $\omega_{i}(t)\in[\underline\omega_{i},\bar\omega_{i}]$ at any
  $t\geqslant t_{0}$.
\end{proof}

\begin{remark}\longthmtitle{Control framework without bottom
    layer}\label{rmk:no-bottom-layer}
  {We have shown in~\cite{YZ-JC:18-cdc1} that even if one only
    implements the first-layer controller (i.e.,
    $\alpha_{MPC}\equiv\zeros_{n}$, leading to $\alpha=\alpha_{DF}$),
    the closed-loop system still meets all requirements except for the
    economic cooperation. Such a lack of cooperation can be seen from
    two aspects. First, since $\alpha_{DF}$ is only available for
    nodes in $\II{\omega}$, those in $\II{u}\backslash\II{\omega}$
    do not get involved in controlling frequency transients. Second,
    the first-layer control is a non-optimization-based state
    feedback, where each $\alpha_{DF,i}$ with $i\in\II{\omega}$ is
    merely in charge of controlling transient frequency for its own
    node~$i$.}  \oprocend
\end{remark}
  
Since the centralized doubled-layered control scheme already meets
requirements~(i)-(v), in the next section, we deal with the remaining
distributed computation requirement.

\section{Controller decentralization through network
  division}\label{section:distributed-control}

Going over the double-layered design in the previous section, it is
worth noticing that the only component of the controller that requires
global information is the MPC component, all the others being local in
nature.  In this section, we propose a distributed double-layered
controller design that addresses this point.  The general idea is to
split the computation of the MPC component across different regions,
and have each region determine its own MPC component based on its
regional state and regional forecasted power information. 

We split the network into regions so that each controlled node is
contained in exactly one region. Formally,
let 
% \begin{assumption}\longthmtitle{Regional division through induced
%     subgraph}\label{assumption:subgraph-node}
$\{\mathcal{G}_{\beta}=(\mathcal{I}_{\beta},\mathcal{E}_{\beta})
\}_{\beta\in[1,d]_{\naturals}}$ be induced subgraphs of $\mathcal{G}$
such that
\begin{subequations}\label{sube:assu-subgraph}
  \begin{align}
    &\II{u}\subseteq\bigcup_{\beta=1}^{d}\mathcal{I}_{\beta},
    \label{sube:assu-subgraph-1}
    \\
    &\mathcal{I}_{\eta}\bigcap\mathcal{I}_{\beta} \bigcap \II{u}=
    \emptyset,\ \forall \eta,\beta\in [1,d]_{\naturals} \text{ with }
    \eta\neq\beta.\label{sube:assu-subgraph-2}
  \end{align}
\end{subequations}
For each subgraph $\GG_{\beta}$, let
$\II{u}_{\beta}\triangleq\II{u}\bigcap\mathcal{I}_{\beta}$
(resp. $\II{\omega}_{\beta}\triangleq\II{\omega}\bigcap\mathcal{I}_{\beta}$)
denote the collection of controlled node indexes (resp. nodes indexes
with transient frequency requirements) within $\GG_{\beta}$.  Let
$(f_{\beta}, \omega_{\beta},\alpha_{MPC,\beta}) \in
\real^{2|\mathcal{I}_{\beta}|+|\mathcal{E}_{\beta}|}$ be the
collection of states in $\GG_{\beta}$. Let
$p^{fcst}_{t,\beta}:[t,t+\tilde
t]\rightarrow\real^{|\mathcal{I}_{\beta}|}$ be the forecasted power
injection for every node in $\GG_{\beta}$ starting from time $t$ to
$\tilde t$ seconds later. Note that the dynamics of $\GG_{\beta}$ is
not completely determined by
$(f_{\beta},\omega_{\beta},\alpha_{MPC,\beta})$ due to its
interconnection with other parts of the network outside $\GG_{\beta}$
through transmission lines with $i\in\mathcal{I}_{\beta}$ and $j \in
\mathcal{I}\backslash\mathcal{I}_{\beta}$ (equivalently, with $(i,j)
\in \mathcal{E}_{\beta}'$).  Instead of considering the flows
$f_{ij}$'s of these transmission lines as states for $\GG_{\beta}$, we
model them as exogenous power injections. Formally, denote for every
$i\in\mathcal{I}_{\beta}$,
% \vspace{-0.8cm}
\begin{align}\label{eqn:pfcst-f}
  p_{t,\beta,i}^{fcst,f}(\tau)\triangleq\sum_{\substack{j:j\rightarrow
      i\\(j,i)\in\mathcal{E}'_{\beta}}}f_{ji}(t)-\sum_{\substack{j:i\rightarrow
      j\\(i,j)\in\mathcal{E}'_{\beta}}}f_{ij}(t),\;\forall \tau\in
  [t,t+\tilde t]
\end{align}
%
% \marginJC{What does the notation $j\rightarrow
%       i$ mean? It is not worth introducing it if we are only using it
%       once, maybe there is another notational solution }
% %
as the forecasted exogenous power injection acting on node~$i$ caused
by transmission lines in $\mathcal{E}_{\beta}'$, where $\left\{ j :
  j\rightarrow i \right\}$ is the shorthand notation for $\left\{ j :
  j\in\mathcal{N}(i) \text{ and $j$ is the positive end of $(i,j)$}
\right\}$.  For simplicity, here we take the forecasted value starting
from time $t$ to be constant within the time interval $[t,t+\tilde
t]$.
%
%\marginJC{And exact at time $t$, no?}
%
Denote by $p^{fcst}_{t,\beta}:[t,t+\tilde
t]\rightarrow\real^{|\mathcal{I}_{\beta}|}$ the collection of all
$p^{fcst}_{t,i}$'s with $i\in\mathcal{I}_{\beta}$, and let $\bar
p_{t,\beta}^{fcst}\triangleq p^{fcst}_{t,\beta}+p^{fcst,f}_{t,\beta}$
be the overall forecasted power injection.  We illustrate these
definitions in an example.

\begin{example}\longthmtitle{A network division in IEEE 39-bus
    network}\label{ex:IEEE-39-division} 
  Fig.~\ref{fig:IEEE39bus} shows a network division example with
  $\II{\omega}=\{30,31,32,37\}$ and $\II{u}=\{3,7,25,30,31,32,37\}$.
  The set $\II{u}$ consists of all nodes in $\II{\omega}$ and all nodes
  capable of adjusting their loads and within two hops of a node in
  $\II{\omega}$. We split the network into three regions ($d=3$)
  satisfying~\eqref{sube:assu-subgraph}. Each region
  $\mathcal{G}_{\beta}$ contains the two-hop neighborhood for every
  node in $\II{\omega}_{\beta}$. We denote by $\mathcal{G}_{1}$ the
  upper left region in Fig.~\ref{fig:IEEE39bus} and use it to
  illustrate related definitions. In $\mathcal{G}_{1}$, one has
  $\mathcal{I}_{1}=\{1,2,3,25,26,30,37\}$,
  $\II{\omega}_{1}=\{30,37\}$, $\II{u}_{1}=\{3,25,30,37\}$,
  $\mathcal{E}_{1}=\{(1,2),(2,30),(2,25),(3,25),(25,37),(26,37)\}$,
  and
  $\mathcal{E}_{1}'=\{(1,39),(3,4),(3,18),(26,27),(26,28),(26,29)\}$. For
  every $i\in\mathcal{I}_{1}$, one can compute
  $p_{t,\beta,i}^{fcst,f}$ by~\eqref{eqn:pfcst-f}, and it is easy to
  see that $p_{t,\beta,i}^{fcst,f}\equiv0$ for $i\in\{2,3,25,30,37\}$,
  as these nodes are not ends of any edge in $\mathcal{E}_{1}'$.
  \oprocend
\end{example}

\begin{figure}[htb]
  \vspace{-0.5cm}
  \centering%
  \includegraphics[width=1\linewidth]{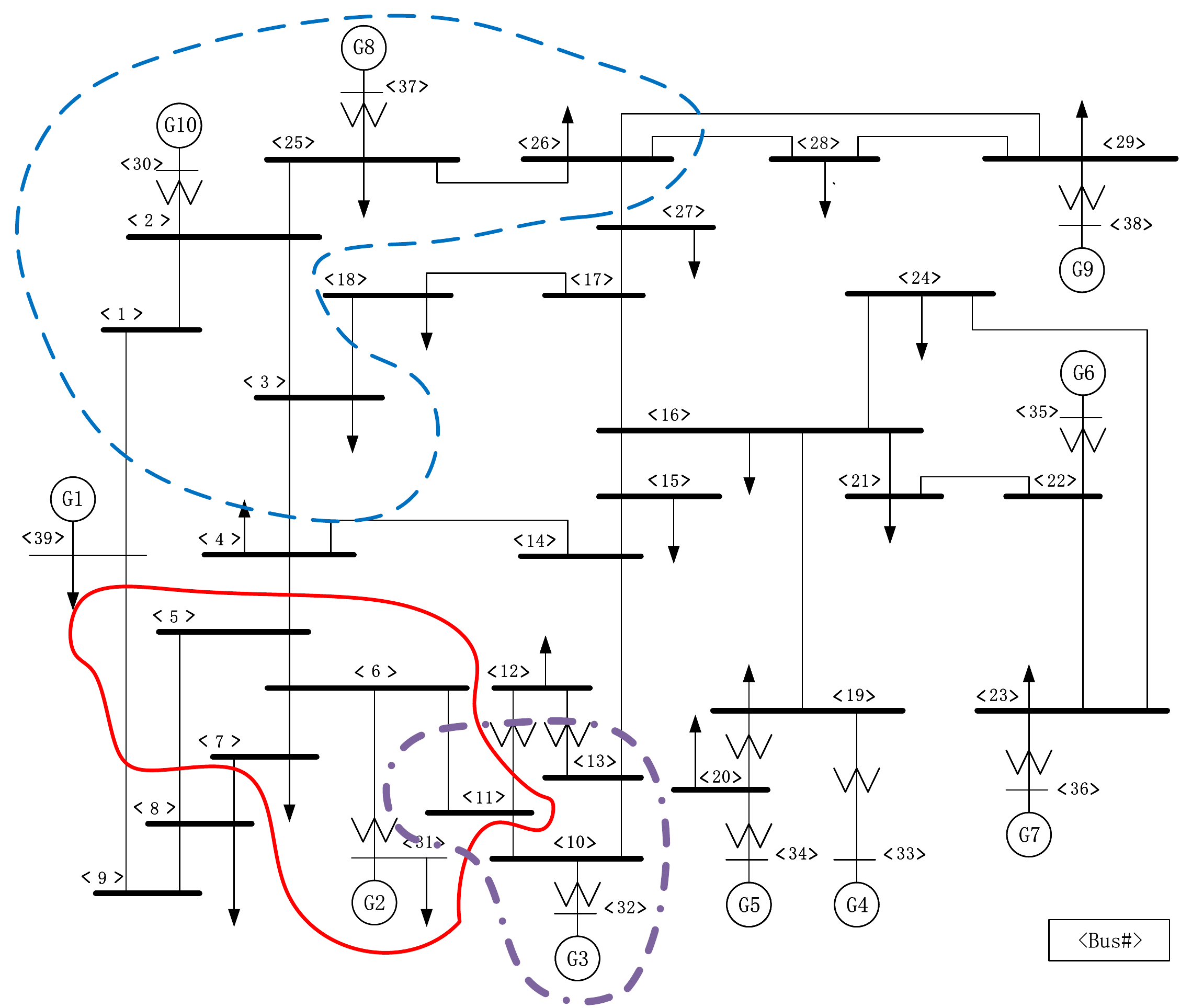}
  \vspace{-0.5cm}
  \caption{IEEE 39-bus power network.}\label{fig:IEEE39bus}
  \vspace*{-1.5ex}
\end{figure}

The key idea of designing the distributed MPC component is to consider
each region as an single network and separately implement the
centralized MPC on it. Formally, for every
$\beta\in[1,d]_{\naturals}$, let
$\{\Delta^{j}_{\beta}\}_{j\in\naturals}$ be its sampling sequence. For
every $i\in\II{u}$, select the unique $\beta$ such that
$i\in\mathcal{I}_{\beta}$, and at every
$t\in[\Delta^{j}_{\beta},\Delta^{j+1}_{\beta})$ with $j\in\naturals$,
let
\begin{align}\label{eqn:distributed-uMPC}
  u_{MPC,i}(t)& \notag
  \\
  &\hspace{-1cm}=\hat
  u^{*}_{i}(\mathcal{G}_{\beta},\II{u}_{\beta},\II{\omega}_{\beta},
  \bar p^{fcst}_{\Delta^{j},\beta},
  f_{\beta}(\Delta^{j}),\omega_{\beta}(\Delta^{j}),\alpha_{\beta}(\Delta^{j})).
\end{align}
Compared to the centralized MPC in~\eqref{eqn:uMPC}, the distributed
version~\eqref{eqn:distributed-uMPC} transforms all global
information, including network topology, forecasted power injection
and system state, into local information. Their structural difference
is that, in the distributed MPC, the overall forecasted power
injection $\bar p^{fcst}_{\Delta^{j},\beta}$ includes an additional
term~\eqref{eqn:pfcst-f} to account for the interconnected dynamics
between the region of interest and the rest of the network.  Next, we
characterize the closed-loop stability and performance of the system
under the distributed controller.

\begin{proposition}\longthmtitle{Distributed double-layered control
    with stability and frequency
    guarantee}\label{prop:two-layer-control-dis} 
  Under Assumption~\ref{assumption:finite-convergence} and assume that
  $\epsilon_{i}T_{i}<1$ for every $i\in\II{u}$, system~\eqref{eqn:compact-form} with controller defined by~\eqref{eqn:two-layer},~\eqref{eqn:stability-filter},~\eqref{eqn:lp-filter},
  \eqref{eqn:second-layer-control}, and~\eqref{eqn:distributed-uMPC} meets requirements~(i)-(iii). Furthermore, $\alpha(t)$, $\alpha_{MPC}(t)$, and $\alpha_{DF}(t)$ converge to $\zeros_{n}$ as $t\rightarrow \infty$.
%   \begin{enumerate}
%   \item\label{item:convergence-dis} The state
%     $(f(t),\omega(t),\alpha_{MPC}(t))$ globally converges to
%     $(f_{\infty},\zeros_{n},\zeros_{n})$.
%   \item\label{item:invariance-dis} For any $i\in\II{\omega}$, if
%     $\omega_{i}(0)\in[\underline\omega_{i},\bar\omega_{i}]$, then
%     $\omega_{i}(t)\in[\underline\omega_{i},\bar\omega_{i}]$ for every
%     $t>0$.
%   \item\label{item:attracitivity-dis} For any $i\in\II{\omega}$, if
%     $\omega_{i}(0)\nin[\underline\omega_{i},\bar\omega_{i}]$, then
%     there exists $t_{0}>0$ such that
%     $\omega_{i}(t)\in[\underline\omega_{i},\bar\omega_{i}]$ for every
%     $t\geqslant t_{0}$.
%   \end{enumerate}
\end{proposition}
%
%\marginJC{Same comment as for the theorem above, seems like we're
%  re-stating the requirements again, no?}
%

The proof is exactly the same as that of
Theorem~\ref{thm:two-layer-control}, and is  omitted.

% \marginJC{We should comment on the fact that, even though the
%   interconnection flows between regions are constants (and hence have
%   errors), we still get all the properties? What have we given up?
%   From the exposition, it seems nothing, when in reality we've given
%   up some optimality. We should comment on that -- and this goes back
%   to my point earlier that the ``economic cooperation'' requirement
%   needs some form of quantification!}
%   \marginy{I'd comment on this in the journal-now its more than 8 pages.}

\section{Simulations}\label{sec:simulations}

Here, we illustrate the performance of the distributed controller in
the IEEE 39-bus power network described in
Fig.~\ref{fig:IEEE39bus}. All parameters in the network
model~\eqref{eqn:compact-form} are taken from the Power System
Toolbox~\cite{KWC-JC-GR:09}. We assign a small rotational inertia
$M_{i}=0.1$ to all non-generator nodes for simplicity. Let
$\bar\omega_{i}=-\underline\omega_{i}=0.2Hz$, so that the safe
frequency region is $[59.8Hz,\ 60.2Hz]$.  To set up the distributed
MPC component~\eqref{eqn:distributed-uMPC}, we select $\tilde t=2s$
and $T=0.02$, so that the predicted step $N=100$; $\epsilon_{i}=1.9$
and $T_{i}=0.5$ for every $i\in\II{u}$; $c_{i}=1$ if
$i\in\II{\omega}$, while $c_{i}=4$ if
$i\in\II{u}\backslash\II{\omega}$; $d=100$;
$\{\Delta^{j}_{\beta}\}_{j\in\naturals}=\{j\}_{j\in\naturals}$ for
every $\beta\in[1,d]_{\naturals}$, i.e., in each region, the MPC
component samples and updates its output every $1s$;
$p^{fcst}_{t}(\tau)=p(\tau)$ for every $\tau\in[t,t+\tilde t]$, i.e.,
the forecasted power injection is precise. To set up the top layer
controller~\eqref{eqn:second-layer-control}, let
$\bar\gamma_{i}=\underline\gamma_{i}=1$ and
$\bar\omega_{i}^{\text{thr}}=-\underline\omega_{i}^{\text{thr}}=0.1Hz$
for every $i\in\II{\omega}$.

% \marginJC{Why don't we simulate with error in forecasted power
%   injection that grows with $\tau-t$?}
% \marginy{We cound do it in journal. Also, I dont think we have enough time to do it today.}

We first show that the distributed controller defined
by~\eqref{eqn:two-layer},~\eqref{eqn:stability-filter},~\eqref{eqn:lp-filter},
\eqref{eqn:second-layer-control}, and~\eqref{eqn:distributed-uMPC} is
able to maintain the targeted nodal frequencies within the safe region
without changing the open loop equilibrium. We disturb all non-generator nodes by
some time-varying power injections. In detail, for every
$i\in[1,29]_{\naturals}$, let $p_{i}(t)=(1+\delta(t))p_{i}(0)$, where
\begin{align*}
  \delta(t)=
  \begin{cases}
    0.2\sin(\pi t\slash50) & \hspace{0.5cm}\text{if $0\leqslant t\leqslant 25$}
    \\
    0.2& \hspace{0.5cm}\text{if $ 25< t\leqslant 125$}
    \\
    0.2\sin(\pi (t-100)\slash50)& \hspace{0.5cm}\text{if $ 125< t\leqslant 150$}
\\
    0& \hspace{0.5cm}\text{if $ 150< t$}
  \end{cases}
\end{align*}
% \begin{figure}[htb]
%   \centering%
%   \includegraphics[width=0.5\linewidth]{epsfiles/IEEE39-power-injection.png}
%   \caption{Evolution of power injection. The trajectory has two
%     fast-varying fluctuations, with a constant deviation
%     inbetween. }\label{fig:power-injection}
%   \vspace*{-1ex}
% \end{figure}
% Fig.~\ref{fig:power-injection} illustrates the power injection. 
The deviation term $\delta(t)p_{i}(0)$ first drops down at a relatively
fast rate and then remains steady for a long time period, finally
converges to 0. We have chosen this scenario to test the capability of
the controller against both fast and slow time-varying power injection
disturbances.

\begin{figure*}[tbh!]
  \centering
  \subfigure[\label{fig:frequency-response-open-loop}]{\includegraphics[width=.24\linewidth]{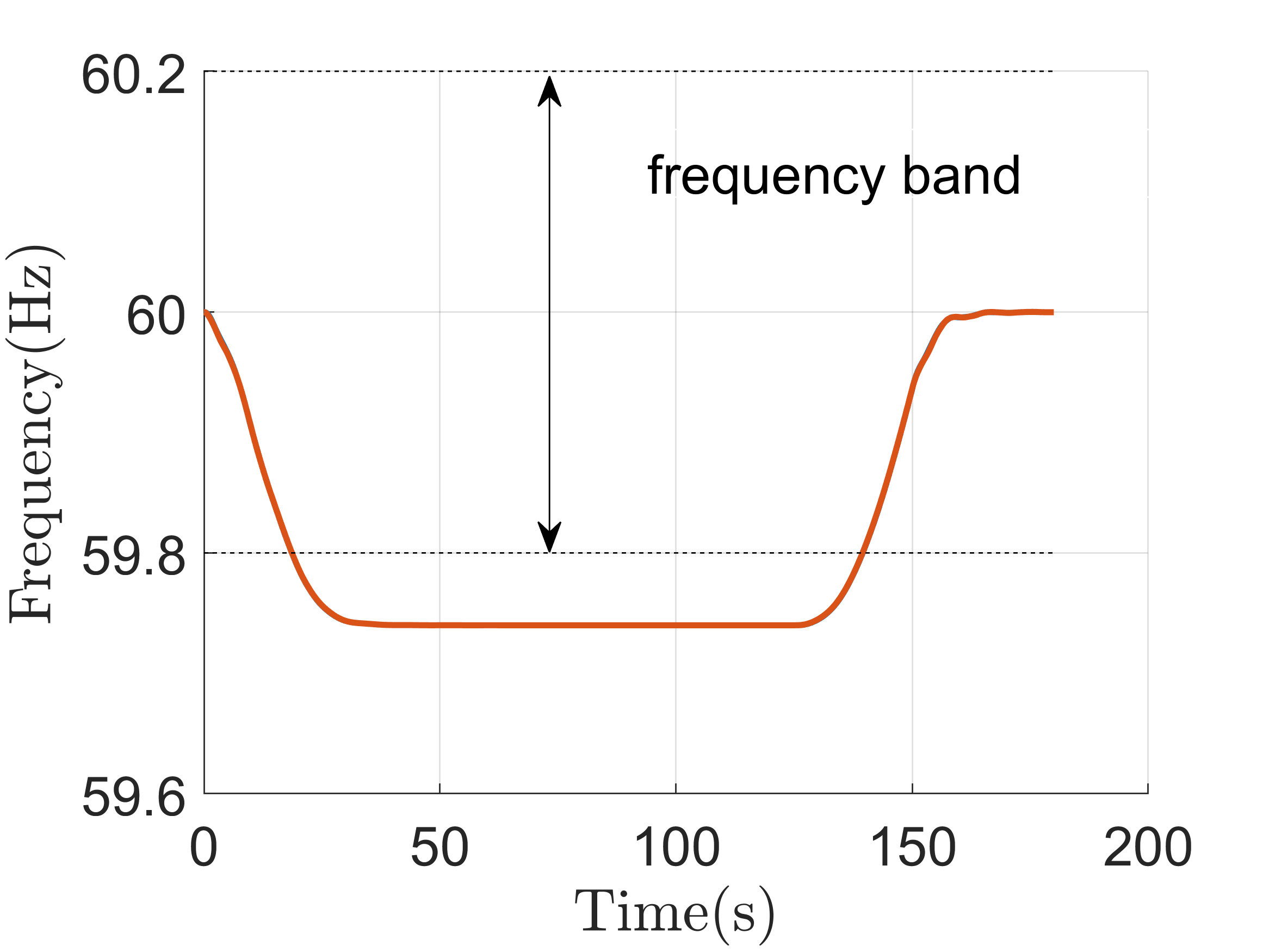}}
  \subfigure[\label{fig:frequency-response-closed-loop}]{\includegraphics[width=.24\linewidth]{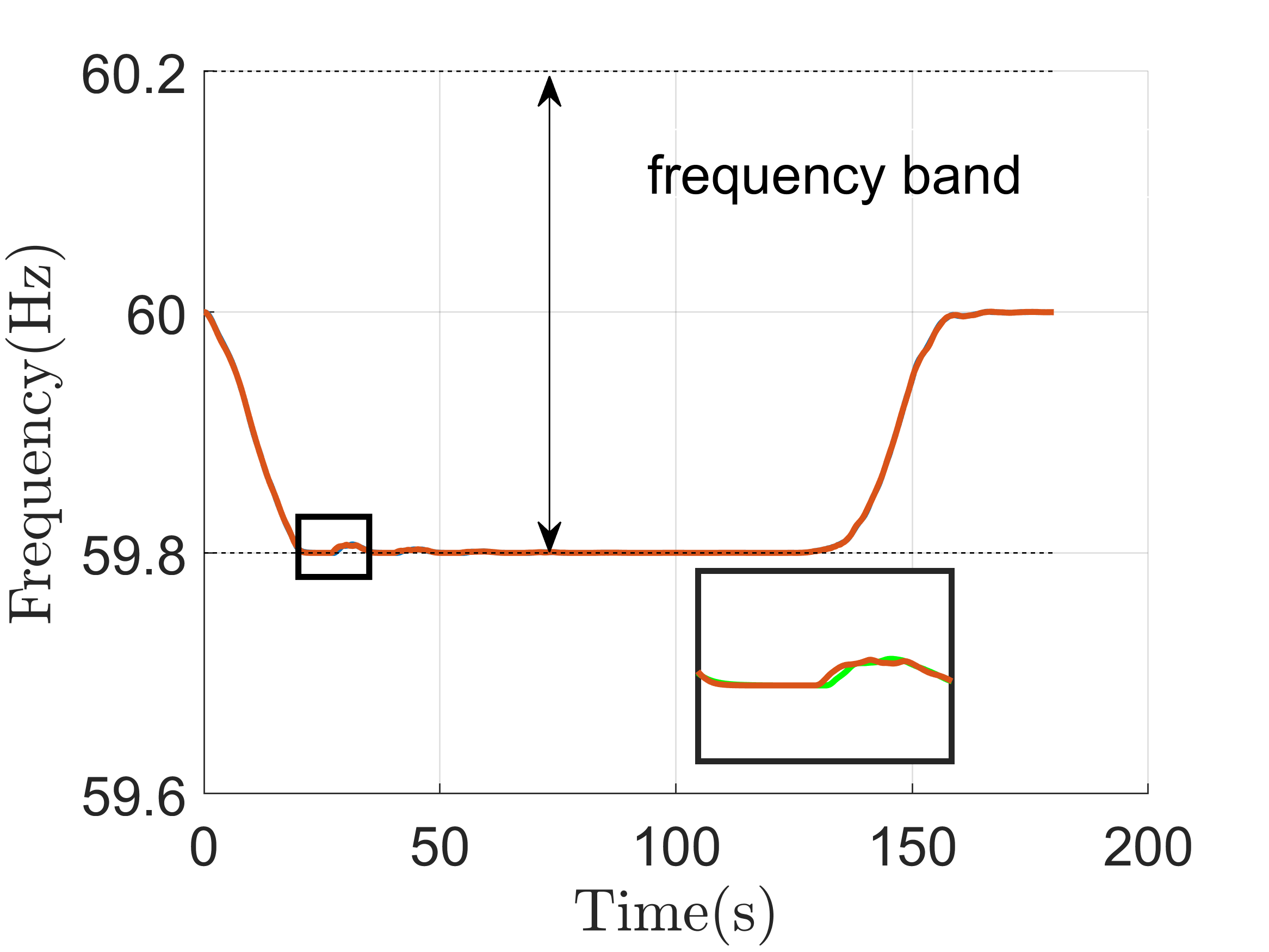}}
  \subfigure[\label{fig:control-response}]{\includegraphics[width=.24\linewidth]{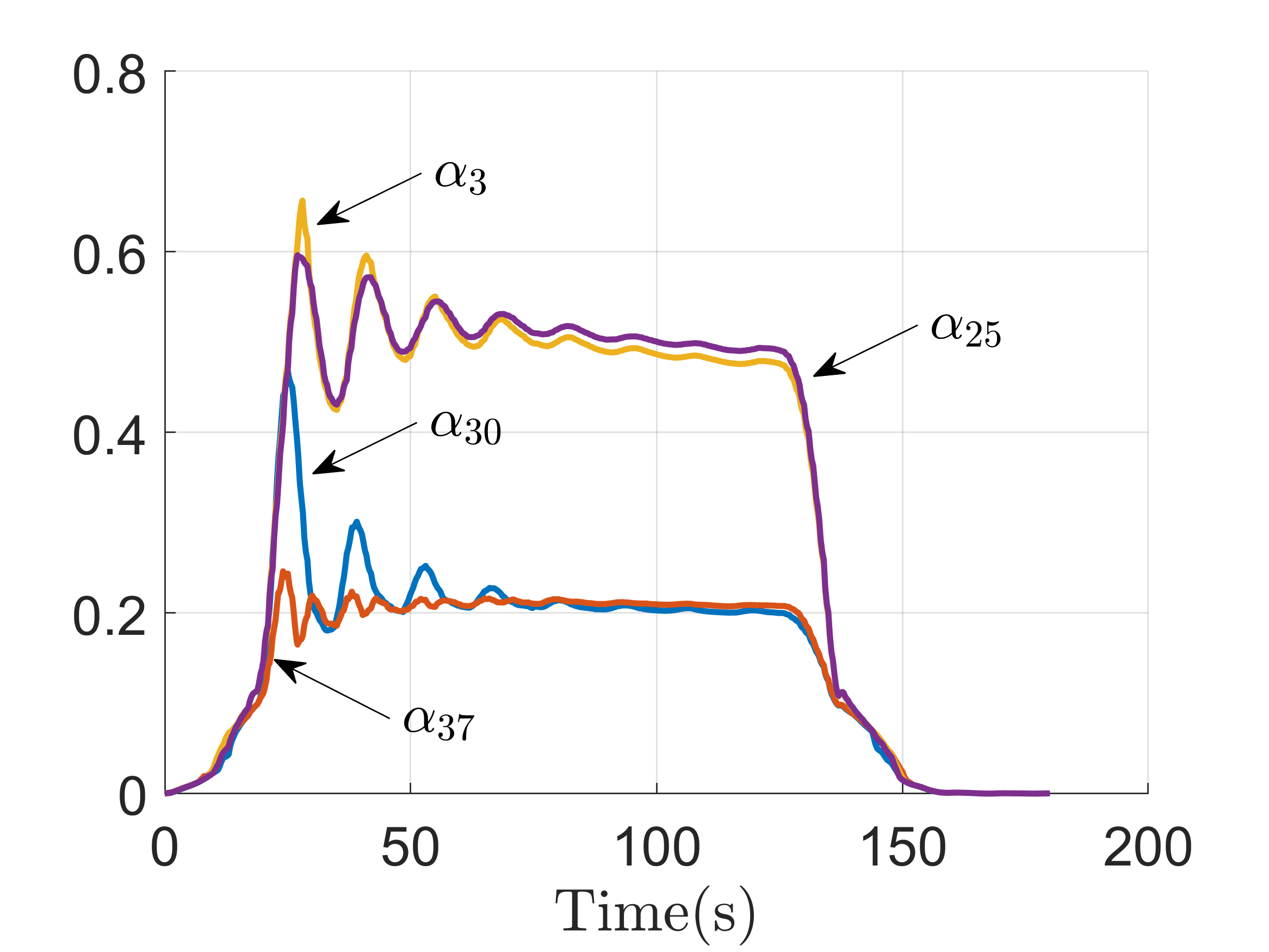}}
    \subfigure[\label{fig:control-response-no-bottom-layer}]{\includegraphics[width=.24\linewidth]{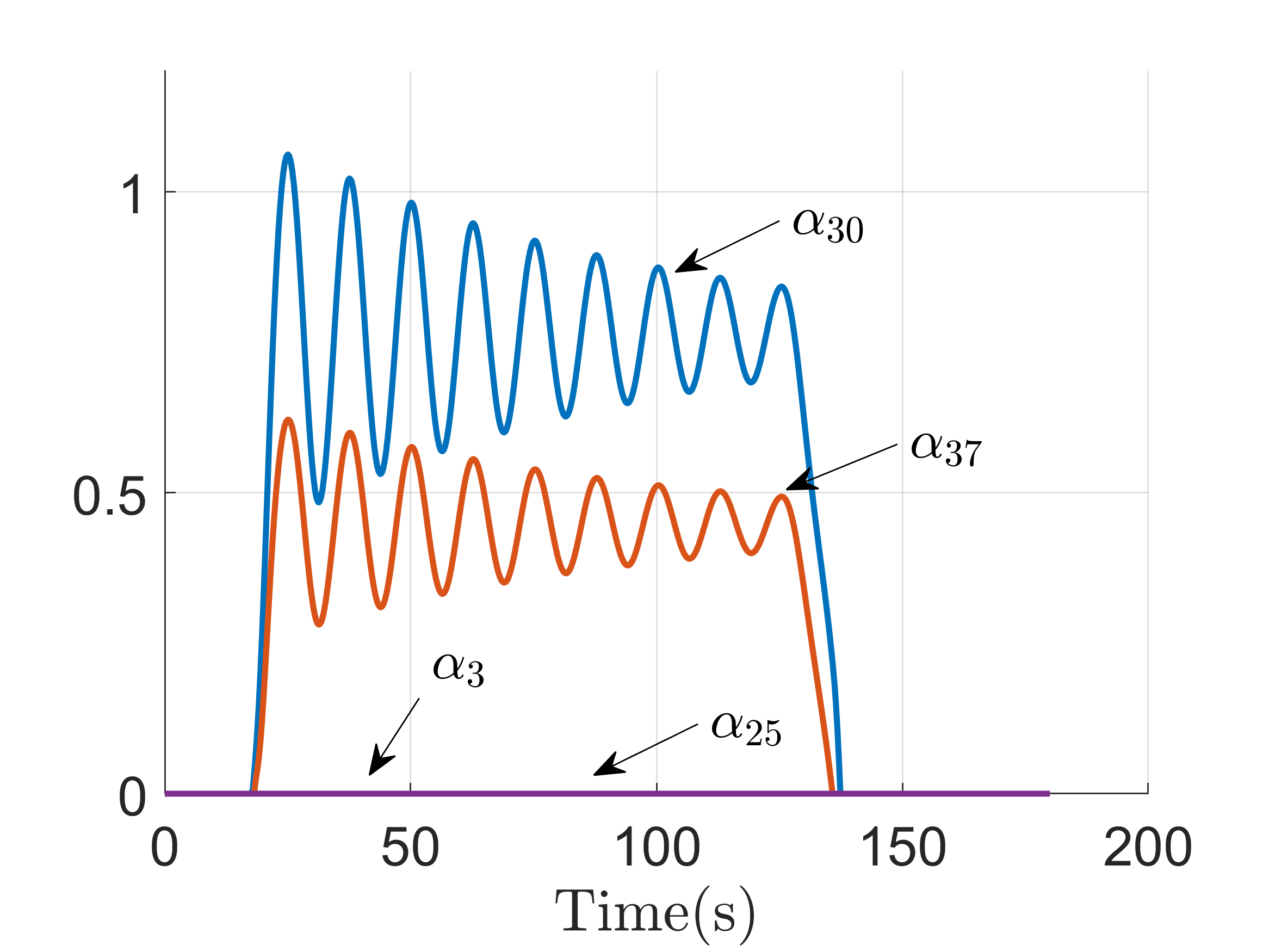}}
    \caption{Frequency and control input trajectories with and without
      distributed transient frequency
      control. Plot~\subref{fig:frequency-response-open-loop} shows
      the open-loop frequency responses at node 30 and 37, both
      exceeding the lower safe bound. With the distributed control, in
      plot~\subref{fig:frequency-response-closed-loop}, both stay
      inside the safe region. Plot~\subref{fig:control-response} shows
      the corresponding control trajectories, where { overall control
        cost 
        $\sum_{i=3,25,30,37}\int_{0}^{200}c_{i}\alpha^{2}_{i}(\tau)\text{d}\tau=97.8$.}
      {Plot~\subref{fig:control-response-no-bottom-layer} shows the
        control trajectories with bottom layer disabled, where the
        overall control cost is 363.5.}}\label{fig:trajectories}
\vspace*{-1.5ex}
\end{figure*}

% \marginJC{Yifu, to be fully fair with the comparison, can you compute
%   the total control effort along time for plot (c) and plot (d), and
%   put that information in the caption?}

For simplicity, in the following we focus on the state and control
input trajectories in the left-top region in Fig.~\ref{fig:IEEE39bus}.
Fig.~\ref{fig:trajectories}\subref{fig:frequency-response-open-loop}
shows the open-loop frequency responses of node 30 and 37, which have
transient frequency requirements.  The two nodes have almost the same
overlapping trajectories that both exceed the safe lower frequency
bound $59.8Hz$. As a comparison, in
Fig.~\ref{fig:trajectories}~\subref{fig:frequency-response-closed-loop},
with the distributed controller, their frequency responses stay within
the safe region, and also gradually come back to $60Hz$ after the
disturbance disappears.  Given the selected coefficients
$c_{3}=c_{25}=1$ and $c_{30}=c_{37}=4$ in the optimization
problem~\eqref{opti:nonlinear}, the controller tends to use
$\alpha_{3}$ and $\alpha_{25}$ more than $\alpha_{30}$ and
$\alpha_{37}$, and this is reflected in the control trajectories in
Fig.~\ref{fig:trajectories}\subref{fig:control-response}. {Fig.~\ref{fig:trajectories}\subref{fig:control-response-no-bottom-layer}
  shows the control trajectories for the non-optimization-based
  controller proposed in~\cite{YZ-JC:18-cdc1}. Compared with those in
  Fig.~\ref{fig:trajectories}\subref{fig:control-response},
  $\alpha_{30}$ and $\alpha_{37}$ do not have a similar trend, and the
  control actions at node $3$ and $25$ have to be disabled,
  cf. Remark~\ref{rmk:no-bottom-layer}.  }

 We then come back to our distributed controller and
further focus on the control signal at node 30 by decomposing
$\alpha_{30}$ into its bottom layer output $\alpha_{MPC,30}$ and top
layer output $\alpha_{DF,30}$. As shown in
Fig.~\ref{fig:trajectories-30}\subref{fig:control-response-30},
$\alpha_{MPC,30}$ is responsible for the larger share in the overall
control signal $\alpha_{30}$, whereas $\alpha_{DF,30}$ only slightly
tunes $\alpha_{30}$. However, if we reduce the penalty $d_{30}$ from
100 to 10, in
Fig.~\ref{fig:trajectories-30}\subref{fig:control-response-30-smaller-penalty},
the dominance of $\alpha_{MPC,30}$ decreases, in accordance with our
discussion in Remark~\ref{rmk:violation-penalty}.
\begin{figure}[tbh!]
  \centering
  \subfigure[\label{fig:control-response-30}]{\includegraphics[width=.48\linewidth]{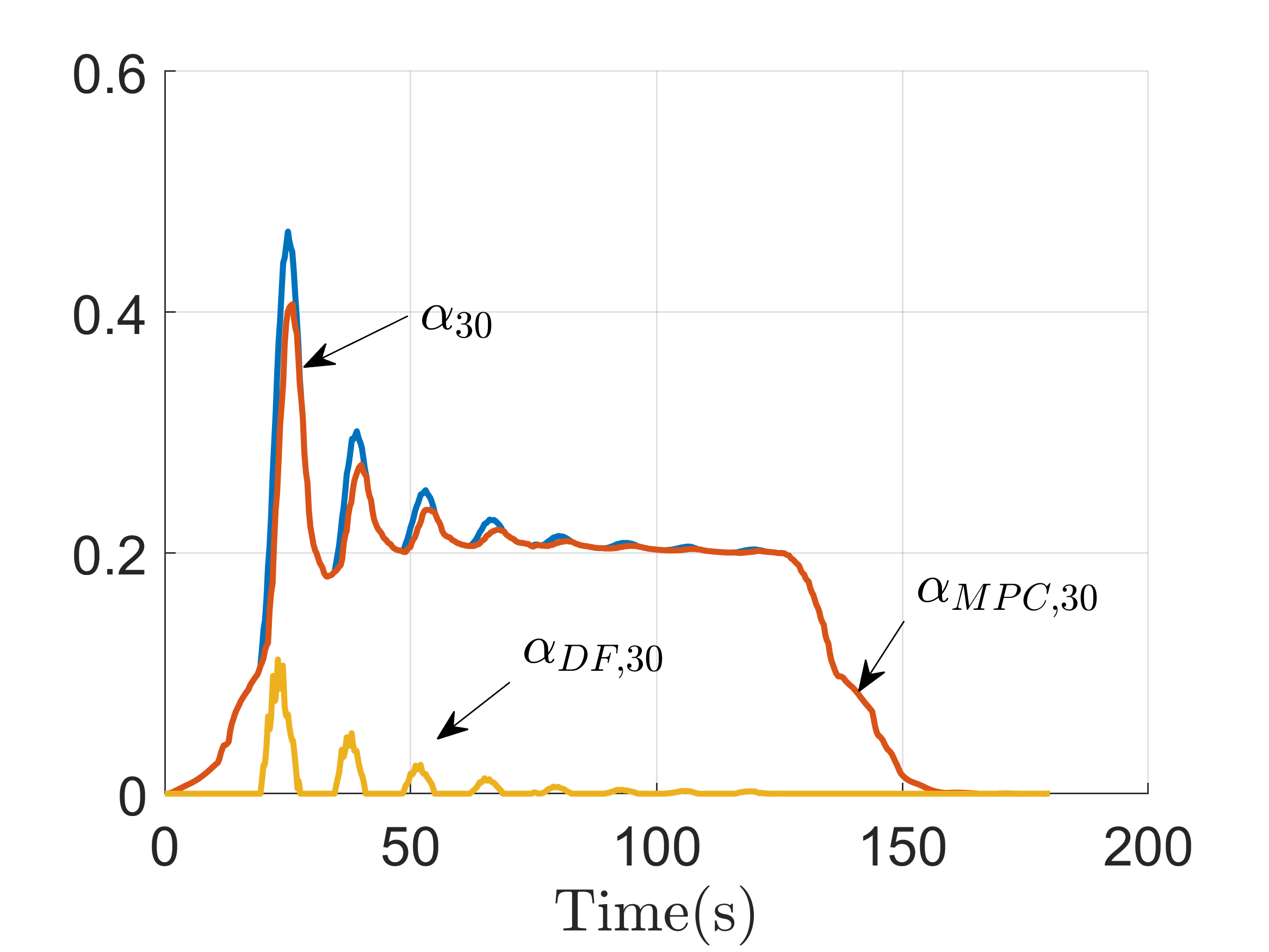}}
  \subfigure[\label{fig:control-response-30-smaller-penalty}]{\includegraphics[width=.48\linewidth]{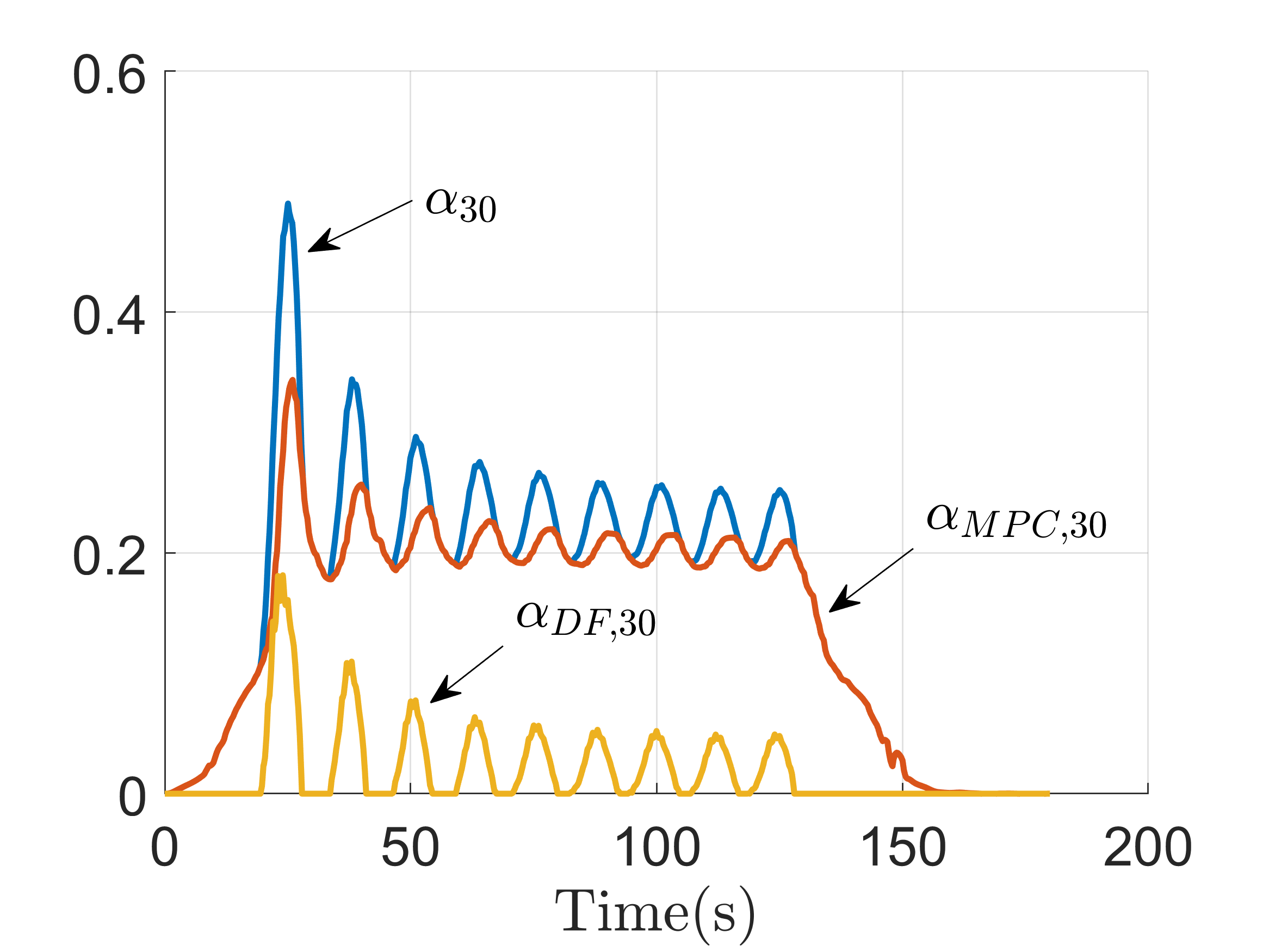}}
  \caption{Control signal decomposition at node 30 with different
    $d_{30}$. In plot~\subref{fig:control-response-30}, with
    $d_{30}=100$, the bottom layer action dominates the total
    input. Such a dominance diminishes as we reduce $d_{30}$ to 10,
    cf.
    plot~\subref{fig:control-response-30-smaller-penalty}.}\label{fig:trajectories-30}
\vspace*{-1.5ex}
\end{figure}
        
Lastly, to verify that the proposed controller meets frequency attractivity requirement,
we consider a case where the initial frequency is outside the safe
region and see how the controller force the frequency back to the
region.  To do so, we disable in the setup above the distributed
controller for the first 30s. In
Fig.~\ref{fig:trajectories-bad-initial}\subref{fig:frequency-response-bad-initial},
one can see that the frequency of node 30 quickly recovers once we
switch on the distributed
controller. Fig.~\ref{fig:trajectories-bad-initial}\subref{fig:control-response-bad-initial}
shows the control signal of node 30. Note that, after transients,
$\alpha_{MPC,30}$ still dominates the overall control signal.

\begin{figure}[tbh!]
  \centering
  \subfigure[\label{fig:frequency-response-bad-initial}]{\includegraphics[width=.48\linewidth]{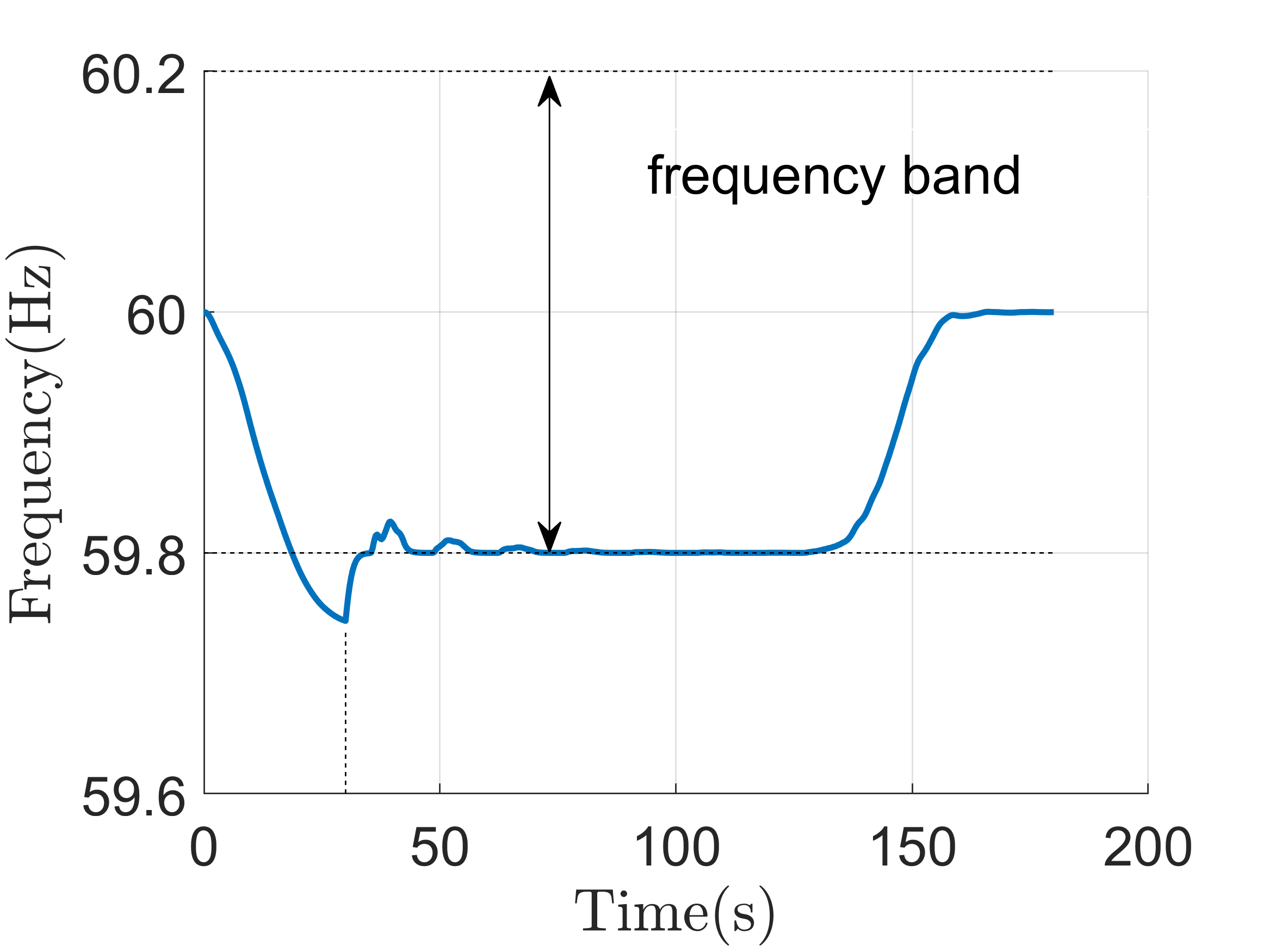}}
  \subfigure[\label{fig:control-response-bad-initial}]{\includegraphics[width=.48\linewidth]{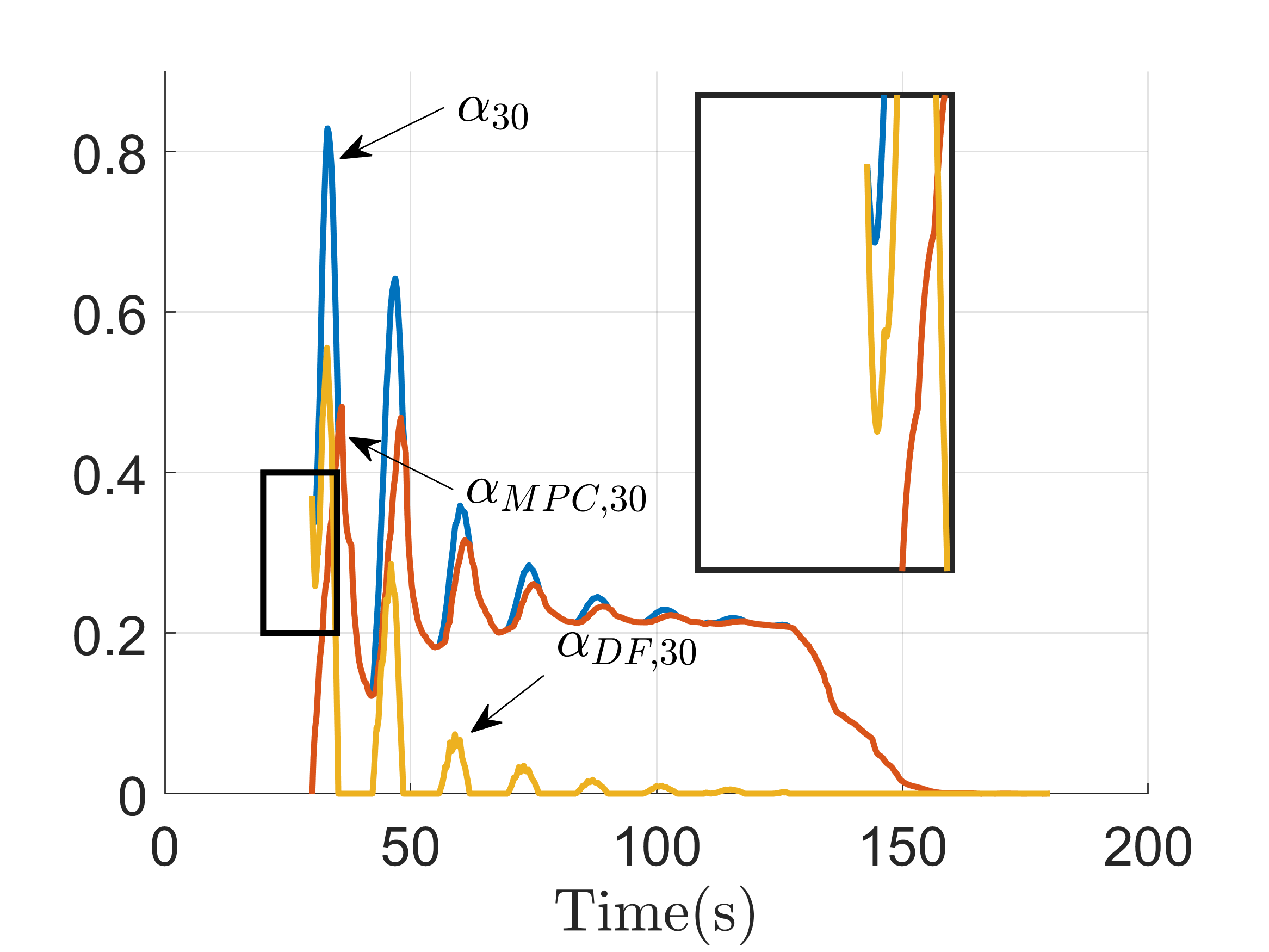}}
  \caption{Frequency and control input trajectories at node 30 with
    controller available after $t=30s$. In
    plot~\subref{fig:frequency-response-bad-initial}, the frequency
    gradually comes back to the safe region after the controller kicks
    in. Plot~\subref{fig:control-response-bad-initial} shows the
    control signals.}\label{fig:trajectories-bad-initial}
\vspace*{-1.5ex}
\end{figure}

% \marginJC{We don't compare our approach in this paper against
%   anything, and I'd assume this could be a source of complaint from
%   the reviewers (e.g., is the complexity of the design justified by
%   what we achieve?). We could do at least one, or all, of the
%   following: (i) compare against [9] to show extent to which
%   ``economic cooperation'' is achieved; (ii) compare against
%   double-layered centralized design to see how much optimality is
%   given up by distributed deisgn: (iii) compare against [10] to see
%   how much optimality traded for implementability.}
%   \marginy{I did (i) in this papaer. We can do (ii) in the journal. (iii) is hard to do due to long simulation time.}

\section{Conclusions}
We have proposed a distributed transient frequency control framework
for power networks that preserves the asymptotic stability of the
network and at the same time, guarantees safe frequency interval
invariance and attractivity for targeted nodes. The controller
possesses a double-layered structure, with the bottom layer
periodically sampling the state and allocating control signals over a
local region in a receding horizon fashion.  The top layer slightly
tunes the bottom layer signal in order to provably enforce frequency
invariance and attractivity guarantees. Implemented over a network
partition, both layers rely on local state and power injection
information. Future work will investigate the extension of the results
to nonlinear swing dynamics, the optimal selection of sampling
sequences in the bottom layer control design, the analysis of the
performance trade-offs of the parameter selections, and the designs of
distributed control schemes that do not rely on network partitions.

%% FOR FINAL VERSION
%\section*{Acknowledgments}
% This work was supported by NSF Award CNS-1446891 and AFOSR Award
% FA9550-15-1-0108.
% 

\bibliographystyle{IEEEtran}%
\bibliography{alias,JC,Main,Main-add}
        
\end{document}